\documentclass[sigconf,authorversion,nonacm]{acmart}

\usepackage{algorithm}
\usepackage[noend]{algpseudocode}

\usepackage{amsmath}

\usepackage{caption}
\usepackage{subcaption}

\usepackage{makecell}

\let\emptyset\varnothing

\usepackage{enumitem}
\usepackage{multirow}

\usepackage{balance}

\usepackage{xcolor}

\newcommand{\revised}[1]{{#1}} 

\def\BibTeX{{\rm B\kern-.05em{\sc i\kern-.025em b}\kern-.08em
    T\kern-.1667em\lower.7ex\hbox{E}\kern-.125emX}}

\usepackage{booktabs} 

\setcopyright{rightsretained}

\acmDOI{}

\acmISBN{978-3-89318-088-2}

\acmConference[EDBT 2023]{26th International Conference on Extending Database Technology (EDBT)}{28th March-31st March, 2023}{Ioannina, Greece}
\acmYear{2023}

\begin{document}

\title{Interactive Set Discovery}

\author{Arif Hasnat}
\affiliation{%
  \institution{University of Alberta}
  \city{Edmonton}
  \country{Canada}}
\email{hasnat@ualberta.ca}

\author{Davood Rafiei}
\affiliation{%
  \institution{University of Alberta}
  \city{Edmonton}
  \country{Canada}}
\email{drafiei@ualberta.ca}


\begin{abstract}
We study the problem of set discovery where given a few example tuples of a desired set, we want to find the set 
in a closed collection of sets.
A challenge is that the example tuples may not uniquely identify a set, and a large number of candidate sets may be returned.
Our focus is on interactive exploration to set discovery where additional example tuples from the candidate sets are shown and the user either accepts or rejects them as members of the target set.
The goal is to find the target set with the least number of user interactions. 
The problem is cast as an optimization problem where we want to find a decision tree that can guide the search to the target set with the least number of questions to be answered by the user. We propose a general algorithm, capable of reaching an optimal solution, and two variations of it that strike a balance between the quality of a solution and the running time. We also propose a novel pruning strategy that safely reduces the search space without introducing false negatives.
\revised{
Our extensive evaluation on both real and synthetic data show that our approach is effective, comparable to or improving upon SOTA, while our pruning strategy reduces the running time of the search algorithms by 2-5 orders of magnitude.}
\end{abstract}

\maketitle

\section{Introduction}
Consider a patient walking to a clinic and being greeted by a machine who does the triage. The patient types headache, nausea and fatigue as symptoms, and the machine checks its database of disease cases and finds over thousands matching each symptom and over hundreds matching all three.
What are the best ways of narrowing down the cases? What are the next few questions the machine asks?

\revised{
Many database interfaces behave in a similar fashion in that there is a large collection of sets and the user is searching for a particular set in the collection. In SQL interfaces, in particular, queries express propositions over \textit{sets of tuples} where the grouping of tuples into sets (e.g. customers who live in Toronto vs those who do not) is \textit{inherent} in queries, and the user is forced to precisely express all propositions in a query to describe the target set.
However, writing SQL queries is a challenging task for many do-it-yourself scientists and professionals. For example, both the Sloan Digital Sky Survey Project~\cite{SDSS} and the SQL-Share Project~\cite{SQLShare,howe2011database,jain2016sqlshare} allow scientists to query their data using SQL, but not many scientists using these projects are expected to know SQL.

In example-based query interfaces~\cite{Weiss-reverse-engineering,Li2015,mottin2017new}, the user provides a set of tuples that are expected to be in the target set (as positive examples) and a set of tuples that are not (as negative examples), and a query that satisfies the given constraints is suggested. The user may adjust the query or provide additional tuples, as positive or negative examples, in an interactive fashion until a precise query expression describing the target set is found.

In browsing-based interfaces, a repository of queries are maintained and the users can browse the repository for similar queries that can be reused with small or no changes~\cite{khoussainova2011session,zhang2019mining}. For example, both the Sloan Digital Sky Survey and the SQL-Share projects keep popular user queries and support searches over those queries. Similarly, in query recommendation, past queries that are recorded in a query log may be recommended as a whole or in part, based on the query fragments that are typed~\cite{chatzopoulou2009query,milo2018next}. 
If each query in the repository is treated as a set of tuples (e.g. the result of the query applied to a database instance), then the recommendation engine is searching for a set in the collection.

In all aforementioned cases, there is a collection of sets or queries and the user is searching for a particular set in the collection. 
We study an interactive exploration approach to set discovery where example tuples from the candidate sets are shown, and the user either accepts or rejects those tuples as members of the target set. 

The number of interactions does not depend on the number of tuples in the database and is only a function of the number of sets (e.g. see Figures~\ref{fig:effects-set-size} and \ref{fig:effects-collection-size}). In many cases, the number of sets that are similar to a target set is expected to be small. For example, Zhong et al.~\cite{zhong-testSuite} show that generating 94 neighbour queries on average for a target SQL query is sufficient to distinguish queries that are different from the target query. If $k$ denotes the number of sets that are similar to a target set, the number of interactions is $k-1$ in the worst cases and closer to $log k$ in most cases.
As each interaction with the user has a cost, we want to retrieve the target set with the least number of interactions.
Relevant research questions are:
(1) what exploration strategies may be used, and how efficient are those strategies?
(2) how long does an exploration take and what factors (e.g., set sizes, overlaps, etc.) do affect the exploration time?
(3) how may the sets be organized to support an efficient exploration?}

\noindent{\bf Problem statement}
\revised{
Informally, set discovery is an interactive process that starts with an initial membership question posed to the user and continues with follow-up questions based on the user's answers. The problem is how to select the next question such that the number of interactions is minimized.
More formally,}
given a collection $C$ of unique sets and an initial set $I$, which includes a subset of the user's desired set, the goal is to find a target set $G$ in $C$ such that $I \subseteq G$. We want to narrow down the search through interactions, i.e. asking the user membership questions about example tuples from $C$, \revised{and we want to find the target set with the minimal number of interactions.} With no user interaction, the problem is under-specified and more than one such set $G$ can contain the elements of $I$ unless $G=I$, in which case a search is meaningless since the user has listed the full target set. 
Also, when $I$ is an empty set, then $G$ is fully identified through interactions with the user.

\noindent{\bf Our approach}
We cast the problem as an optimization with the aim of minimizing the number of questions that the user needs to answer. \revised{The search for a target set is modeled as a tree traversal from the root to a leaf, with each node representing a step of the exploration where the user is given a question that can reduce the number of candidate sets. 
The problem of optimal decision tree construction is NP-hard, and there are strong results on the non-approximability of the problem (see Sec.~\ref{sec:relatedwork}). Our work improves upon a SOTA approximate algorithm, in terms of the quality of the tree, while significantly reducing the size of the search space using strongly effective pruning strategies. 
The tree construction is done \textit{online} in an incremental fashion with questions answered; we also discuss a strategy which aims to reduce the online search cost with an \textit{offline} tree construction.} Assuming all candidate sets in $C$ being equally likely to be the target set $G$, we consider two exploration scenarios: (1) \textit{average-case} where the average number of questions over all possible target sets is minimized, and (2) \textit{worst-case} where the maximum number of questions over all possible target sets is minimized. Our algorithms are general and work under both exploration scenarios.

\noindent\textbf{Contributions}
Our contributions can be summarized as follows:
\begin{itemize}
\item We formalize interactive set discovery as an optimization problem, minimizing the number of questions posed to users.
\item We propose cost functions to characterize the quality of a decision tree for interactive set discovery, in terms of its worst-case and average-case performance, and some lower bounds that are easy to compute but effective in pruning the search space.
\item We propose a pruning strategy, based on our lower bounds, that allows certain choices of entities for decision tree nodes to be safely rejected if there is evidence that it cannot lead to a better tree than one already found.
\item Based on our pruning strategy, we develop an efficient lookahead algorithm that can find near-optimal trees in many cases. We also develop two variations of our lookahead algorithm to further speed up the search process by bounding the number of entities in each step of the search.
\item Through an extensive experimental evaluation, we show that our pruning strategy is effective, reducing the running time by a few orders of magnitude, and that our algorithms outperform competitive approaches from the literature.
\end{itemize}

The rest of the paper is organized as follows. We review the related work in Section~\ref{sec:relatedwork} and present the problem and our formulation in Section~\ref{sec:problem-formulation}. Our algorithms and strategies are discussed in Section~\ref{sec:methodology}, and in Section~\ref{sec:experiments}, we experimentally evaluate their performance. Finally, we conclude with a discussion in Section~\ref{sec:discussion} and provide remarks for future work in Section~\ref{sec:conclusions}.

\section{Related Work}
\label{sec:relatedwork}
Our work is related to the lines of work on (a) example-based query discovery, (b) active learning and interactive query discovery, and (c) cost-efficient decision tree construction.

\subsection{Example-based query discovery}
Our work is related to this line of work in that it can be applied to discover target queries based on example tuples if the candidate queries are known or can be enumerated.
The problem of discovering queries based on examples has its root in QBE~\cite{zloof1975query} and has been lately studied for reverse engineering queries in various domains (e.g., relational~\cite{Tran2009QBO} and
graph data~\cite{mottin2014exemplar,arenas2016reverse}).
On discovering SQL queries, in particular, Tran et al.~\cite{tran2014query} study the problem for select-project queries, and others study project-join queries~\cite{Zhang2013,kalashnikov2018fastqre}.
Weiss et al. \cite{Weiss-reverse-engineering} show that the problem of discovering SPJ from examples is NP-hard when there is a bound on the size of queries. 
An underlying assumption in many of these works is that the queries can be discovered on small instances (where the answer tuples can be easily listed) before being applied to larger instances.
\revised{
Unlike the aforementioned works that focus on a specific query type to keep the complexity of query generation under control, there is no such restriction in our work. Our query discovery is done based on the query output on a sample database, hence any pair of different queries must return different results on the sample instance to be distinguishable. Also we assume the set of queries are given or can be enumerated. There is no other constraint on queries and their complexity.
The approaches on reverse engineering queries are limited in terms of the type of queries they can produce, and these approaches are not applicable when the target is a set and not a query.}

\subsection{Active learning and interactive query discovery}
Related work includes active learning~\cite{settles2009active}, where a model is learned by interacting with a user, and interactive exploration to learn a desired query.
Angluin~\cite{angluin1988queries} shows the polynomial learnability of conjunctions of horn clauses and Abouzied et al. \cite{Abouzied-interactive} show that efficient solutions for a subset of quantified Boolean queries, referred to as role-preserving qhorn queries, are reachable. In both cases, users specify propositions that hold by answering membership questions (e.g. a row is or is not in the answer) and this helps to narrow down the search for candidate queries.
Bonifati et al. \cite{Bonifati14interactiveinference, Bonifati2016} infer join queries and Dimitriadou et al. \cite{Dimitriadou2014} predict conjunctive queries, both
based on interactions in the form of simple yes/no answers about the presence of tuples in the final output.
Li et al. \cite{Li2015} take a sample database and a desired result and generate candidate SPJ queries that produce the result on the sample, using the approach of Tran et al. \cite{tran2014query}. In each  follow-up interaction, the user is provided with a modified database and a collection of query results to choose from, based on which candidate queries are removed until a query emerges.

In all aforementioned works, the search takes place over the space of possible queries that can be generated, and the size of this space is bounded by placing constraints on the the shape of queries. For example, this space in Abouzied et al. \cite{Abouzied-interactive} is the cartesian product of all domains, which means with $m$ Boolean variables, there are $3^m$ possible assignments of constants and don't-care values to those variables and that many propositional logic queries. This space in Bonifati et al. \cite{Bonifati14interactiveinference, Bonifati2016} is all subsets of the Cartesian product of the two tables being joined while limiting the predicates to equijoin, in
Dimitriadou et al. \cite{Dimitriadou2014} is the set of rectangular regions defined by the conjunction of range predicates on numerical data and in Li et al. \cite{Li2015} is a set of conjunctive queries.
In our case, the search takes place over a closed collection of sets, and there is no constraint on the shape of queries that may generate those sets. For example, the sets in one of our datasets are generated using SQL queries with CNF formulas in the where clause and in another dataset by performing union over arbitrary sets.
\revised{
Also unlike active learning where the classes are not explicitly known or not enumerated and achieving 100\% accuracy is out of reach, in our case the sets to be discovered and their boundaries and relationships in terms of overlaps are fully known.}

\subsection{Cost-efficient decision tree construction}
There are some strong results on the nonapproximability of the problem. Sieling~\cite{sieling2008minimization} shows that the problem cannot be approximated up to any constant factor, based on the nonapproximability of Vertex Cover for Cubic graphs and that the problem can be mapped to an optimal decision tree construction. In a much stronger result, Dinue and Steurer~\cite{dinur2014analytical} show that optimal set cover cannot be approximated to $(1-o(1))ln n$ unless $P=NP$, and the same result holds for optimal decision tree construction based on a reduction from set cover~\cite{laurent1976constructing}.
Adler et al.~\cite{Adler2008} propose a greedy algorithm which achieves $(\ln n + 1)$-approximation, by simply choosing an entity at each decision node that most evenly partitions the collection of items. This greedy algorithm sets a strong baseline in terms of the approximability of the problem, and many commonly-used approaches (e.g. Information Gain~\cite{Quinlan1993}, ID3~\cite{Quinlan1986} and C4.5~\cite{Quinlan1993}) are all variations of this 1-step lookahead greedy algorithm (see Sec.~\ref{sec:entity-selection}). 
Esmeir et al.~\cite{Esmeir2004} propose lookahead based algorithms for anytime induction of decision trees by developing k-steps entropy and information gain.
Our proposed algorithm for set discovery improves upon the 1-step lookahead approaches, which are pretty strong baselines, but is 2 to 5 orders of magnitude faster than the k-steps of Esmeir et al., thanks to our powerful pruning strategies. 

\section{Problem Formulation} \label{sec:problem-formulation}
Consider a collection of $n$ candidate sets and a target set in the collection that needs to be identified. Without loss of generality, we assume the sets are all unique; if not, duplicates can be removed without affecting the search task. We want to find the target set through a set of membership questions that the user answers (e.g., Is $A$ in the target set?). At a high level, we want to minimize the number of interactions.

A general approach to the search problem is to construct a decision tree with the candidate sets placed at the leaves and each internal node representing a question. With interactions limited to yes/no membership questions, the decision tree will be a full binary tree with $n$ leaves and $n-1$ internal nodes. The number of such decision trees that can be constructed is huge\footnote{The actual number is the $(n-1)$th Catalan number, i.e.,  $\frac{1}{n}\binom{2(n-1)}{n-1} = \frac{(2(n-1))!}{n!(n-1)!}$.},
and some of those trees are more efficient for finding the target set than others.

Let $m$ denote the size of the universe from which the sets are drawn. For a collection $C$ of finite sets, $m=|\bigcup_{s \in C} s|$. In our presentation, we may refer to the members of the universe as entities, though our approach is applicable to any sets of tuples (e.g., sets of relationships). For a fixed tree shape with $n-1$ internal nodes,  the number of possible placements of $m$ entities or tuples on internal nodes will be $m(m-1)\ldots (m-n+2) = \frac{m!}{(m-n+1)!}$, assuming that each entity appears at most once in the tree. Otherwise, this number is even larger. Searching for an efficient decision tree among all these tree shapes and possible placements of entities on internal nodes is a major computational challenge, and that is the problem studied in this paper.

To alleviate the problem, one may group entities in $C$ into \textit{informative} and \textit{uninformative}. An entity that is either present in all sets in $C$  or none is not informative, since a membership question about that entity does not reduce the search space. The rest of the entities can be considered as informative. Clearly we want to limit our questions to informative entities, and only place those entities on the internal nodes.

\begin{example}
Consider the collection of seven sets, as shown in Fig.~\ref{fig:example-sets}. Entity $a$ is uninformative since it is present in all sets. All the other entities $b, c, ..., k$ are informative. Fig.~\ref{fig:example-trees} shows three possible decision trees that represent the sets in the collection. All the trees are full binary decision trees with 6 internal nodes and 7 leaves. The root node corresponds to all sets of the collection. In Fig.~\ref{fig:example-tree1}, the left branch corresponds to the sub-collection $\{S1, S2, S3\}$ where entity $d$ is present and the right branch corresponds to the sub-collection $\{S4, S5, S6, S7\}$ where $d$ is not present.
Each branch is further broken down based on the presence or absence of entities.
\end{example}

\begin{figure}[!htb]
    \centering
    \begin{tabular}{lll}
        $S1 = \{a, b, c, d\}$ & $S2 = \{a, d, e\}$ & $S3 = \{a, b, c, d, f\}$ \\
        $S4 = \{a, b, c, g, h\}$ & $S5 = \{a, b, h, i\}$ & $S6 = \{a, b, j, k\}$ \\
        $S7 = \{a, b, g\}$ & &
    \end{tabular}
  \caption{A collection of example sets}
  \label{fig:example-sets}
\end{figure}

\begin{figure*}[!htb]
  \centering
    \begin{subfigure}[b]{0.33\linewidth}
        \centering
        \includegraphics[width=\linewidth]{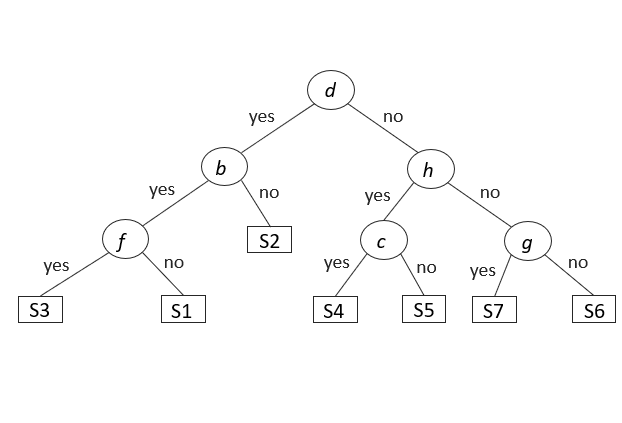}
        \caption{}
        \label{fig:example-tree1}
    \end{subfigure}
    \hfill
    \begin{subfigure}[b]{0.33\linewidth}
        \centering
        \includegraphics[width=\linewidth]{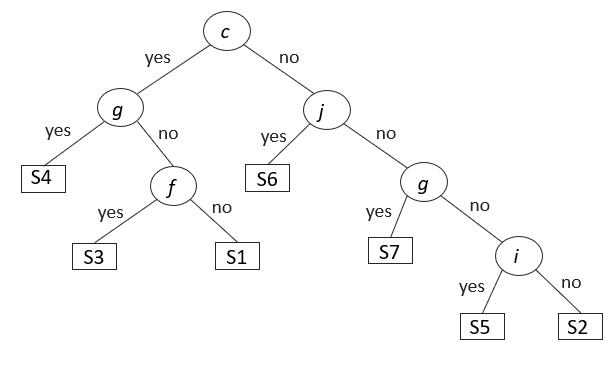}
        \caption{}
         \label{fig:example-tree2}
    \end{subfigure}
    \begin{subfigure}[b]{0.33\linewidth}
        \centering
        \includegraphics[width=\linewidth]{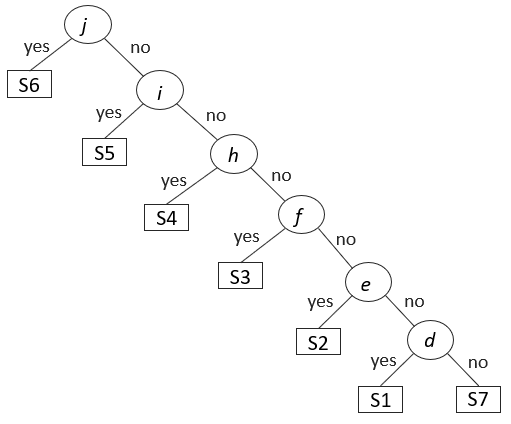}
        \caption{}
         \label{fig:example-tree3}
    \end{subfigure}
  \caption{Example of decision tree representations of the sets in Figure~\ref{fig:example-sets}}
  \label{fig:example-trees}
\end{figure*}

Given a decision tree, the number of questions that are required to find a set is determined by the depth at which the set is placed. For example, in Fig.~\ref{fig:example-tree1}, $S2$ can be detected using two questions whereas one will need three questions to find any other set. Since we do not know the target set in advance, and assuming that all sets are equally likely, the cost of a tree can be defined as the average depth of the leaves which equivalently represents the expected number of questions required to find the target set.

\begin{definition}
\label{def:ad}
Let $T$ be a full binary decision tree over a collection $C$ of unique sets, i.e., $T$ has exactly $|C|$ leaves and each leaf is labelled with a set in $C$. If \textit{depth(s,T)} denote the depth of a set $s$ in $T$, then the cost of $T$ is defined as
\[cost(T) = \frac{\sum_{s \in C}depth(s,T)}{|C|}.\]
\end{definition}

Alternatively, one can also define the cost of a tree as the height\footnote{Here height refers to the depth of the leaf with the longest distance from the root, i.e., the number of questions to be answered to reach the deepest leaf.} of the tree.
For a collection with $n$ unique sets, the height of a full binary decision tree cannot be less than $\lceil\log_2n\rceil$ for $n>0$. This sets a lower bound on the height (H) of an optimal tree, which we refer to as $LB\_H(n)$. The next lemma gives the lower bound on the average depth of the leaves.

\begin{lemma}
Given a collection of $n$ unique sets such that $n > 0$, a lower bound on the average depth of the leaves (AD) of a full binary decision tree representing the collection, denoted as $LB\_AD(n)$, is $\lceil n\log_2n\rceil/n$.
\label{lem:least-ad}
\end{lemma}

\begin{proof}
The average depth of the leaf nodes of a full binary decision tree representing $n$ sets cannot be less than $\log_2n$. Hence, the sum of depth of the $n$ leaf nodes cannot be less than $\lceil n\log_2n\rceil$ since it must be an integer number. Therefore, a lower bound on AD for $n$ unique sets is $\lceil n\log_2n\rceil/n$.
\end{proof}

Now let's examine the trees in Fig.~\ref{fig:example-trees} again. A lower bound on AD of any full binary decision tree representing a collection of 7 sets, according to Lemma~\ref{lem:least-ad}, is $2.857$. The AD of the tree in Fig.~\ref{fig:example-tree1} is $2.857$, Fig.~\ref{fig:example-tree2} is $3.0$ and that of the tree in Fig.~\ref{fig:example-tree3} is $3.857$, hence the first tree is optimal but the others are not.

Given a collection of $n$ unique sets, our goal can be stated as finding a full binary decision tree representation of the collection with the least cost where the cost metric is either AD or H.

\section{Methodology}
\label{sec:methodology}
Given a collection $C$ of unique sets, our set discovery process constructs a decision tree with the sets in $C$ placed at the leaves and the membership questions about entities placed at the internal nodes.
Since there are many possible trees that can be constructed, and some are more efficient than others, our goal is to find a tree that leads to the least number of interactions with the user. Before presenting our algorithms, we develop a few lower bounds on cost, which will be used in pruning the search space of our algorithms.

\subsection{Cost Lower Bounds}
\label{sec:lower-bounds}
Given a collection $C$ of unique sets, two lower bounds on cost (as discussed in Section~\ref{sec:problem-formulation}) are
\begin{equation}\label{eqn:lb-0-ad}
    LB\_{AD}_{0}(C) = \frac{\lceil |C| * \log_2|C|\rceil}{|C|}, and
\end{equation}
\begin{equation}\label{eqn:lb-0-h}
    LB\_H_{0}(C) = \lceil\log_2|C|\rceil
\end{equation}
for cost metrics AD and H respectively. Now consider an 
entity $e$ that partitions $C$ into two sub-collections $C1$ and $C2$. 
Our cost lower bounds, after placing $e$ at the current node of the decision tree, can be written as

\begin{equation}\label{eqn:lb-1-e-ad}
    LB\_AD_{1}(C, e) = \frac{|C1| * LB\_AD_{0}(C1) + |C2| * LB\_AD_{0}(C2)}{|C|} + 1,
\end{equation}	
and
\begin{equation}\label{eqn:lb-1-e-h}
  LB\_H_{1}(C, e) = max(LB\_H_{0}(C1),  LB\_H_{0}(C2)) + 1
\end{equation}
for cost metrics AD and H respectively, where the index `1' in the bounds indicates that the cost is calculated using information available after looking one step ahead, i.e. one level below the current node.
We use the general term $LB$ to refer to any lower bound (including $LB\_AD$ and $LB\_H$) when a distinction in the cost metric is not important.
	
Let $E$ denotes the set of entities in collection $C$. A lower bound on cost over all entities with 1-step look ahead is
\begin{equation}\label{eqn:lb-1}
   LB_{1}(C) = min_{e \epsilon E}LB_{1}(C, e). 
\end{equation}

These definitions can be extended for k-steps lookahead  as
\begin{equation}\label{eqn:lb-k-e-ad}
\begin{split}
    LB\_AD_{k}(C, e) &= \frac{|C1| * LB\_AD_{k-1}(C1) + |C2| * LB\_AD_{k-1}(C2)}{|C|} \\
    &+ 1,\ and
\end{split}
\end{equation}
\begin{equation}\label{eqn:lb-k-e-h}
  LB\_H_{k}(C, e) = max(LB\_H_{k-1}(C1),  LB\_H_{k-1}(C2)) + 1
\end{equation}
for cost metrics AD and H respectively. A lower bound over all entities is
   
\begin{equation}\label{eqn:lb-k}
    LB_{k}(C) = min_{e \epsilon E}LB_{k}(C, e). 
\end{equation} 
   
A desirable property of these lower bounds is their monotonicity and that the cost lower bounds never decrease. This means the lower bounds can get tighter but not looser, as we look more and more steps ahead. The next two lemmas state this more formally.

\begin{lemma}\label{lem:least-cost-comp}
For any collection $C$, $LB_k(C)$ is a monotone non-decreasing function of $k$,
i.e., for non-negative integers  $k1$ and $k2$, if $k2 > k1$, then $LB_{k2}(C) \geq LB_{k1}(C)$. 
\end{lemma}

\begin{proof}
The proof is by induction on $k$. For the basis, 
the lowest possible cost of a binary decision tree on $C$, defined as $LB_{0}(C)$, is calculated assuming that the entity in each node of the tree partitions the sub-collection as evenly as possible. But, $LB_{1}(C)$ is calculated after an actual entity from the collection is assigned to the root node and assuming that all other nodes partitions the sub-collections as evenly as possible. If the entity at the root partitions the collection as evenly as possible then $LB_{1}(C)$ is equal to $LB_{0}(C)$. Otherwise, $LB_{1}(C)$ is greater than $LB_{0}(C)$. 
For the induction step, suppose the claim holds at step $k1$. In each additional step $k2 = k1 + 1$ of the lower bound calculation, an additional level of nodes are assigned with the best entities recursively. If any of those entities does not partition the corresponding sub-collections as evenly as possible then $LB_{k2}(C) > LB_{k1}(C)$, otherwise, $LB_{k2}(C) = LB_{k1}(C)$. Therefore, the statement holds.
\end{proof}

\begin{lemma}\label{lem:least-cost-e-comp}
For any collection $C$ and entity $e$ in the collection, $LB_k(C,e)$ is a monotone non-decreasing function of $k$,
i.e., for positive integers  $k1$ and $k2$, if $k2 > k1$, then $LB_{k2}(C,e) \geq LB_{k1}(C,e)$. 
\end{lemma}

The proof follows the line of reasoning in Lemma~\ref{lem:least-cost-comp}.

\subsection{Entity Selection}\label{sec:entity-selection}
The problem of constructing an optimal binary decision tree, minimizing the cost to discover an unknown target set, is NP-complete~\cite{HYAFIL1976}, hence various greedy strategies have been studied in the literature. In this section, we briefly review these strategies and compare them with ours.

\subsubsection{Most even partitioning}
A greedy approximation algorithm which achieves $(\ln n + 1)$-approximation for the decision tree problem on a collection C with $n$ sets is simply to choose an entity at each internal node that most evenly partitions the collection of sets in that node~\cite{Adler2008}.

\subsubsection{Information gain}
Decision tree construction is a very well-understood process in machine learning and data mining. A popular heuristic used by the decision tree algorithms (such as ID3~\cite{Quinlan1986} and C4.5~\cite{Quinlan1993}) for selecting the next feature or entity is the information gain. The entity with the largest information gain is selected to split the collection. If we treat each set in $C$ as a class and each distinct entity $e$ as a feature, then the information gain of $e$ that partitions $C$ into sub-collections $C1$ and $C2$ can be written as

\begin{equation}\label{eqn:info-gain}
    InfoGain(C, e) = \log_2|C| - \frac{|C1| * \log_2|C1| + |C2| * \log_2|C2|}{|C|}.
\end{equation}

\subsubsection{Indistinguishable pairs}
Another entity selection strategy (used by Roy et al.~\cite{Roy2008}) selects an attribute or entity that minimizes the number of indistinguishable pairs of sets. For an entity $e$ that partitions a collection $C$ into sub-collections $C1$ and $C2$, the number of indistinguishable pairs is given as

\begin{equation}\label{eqn:indg}
   Indg(C, e) = \frac{|C1| * (|C1| - 1) + |C2| * (|C2| - 1)}{2}.
\end{equation}

\subsubsection{Cost lower bound}
Entity selection can be done using our cost lower bound $LB_k$ with $k > 0$, as discussed in Section~\ref{sec:lower-bounds}, by selecting the entity that minimizes the cost lower bound.
This is the strategy we use in this paper because of some of the desirable properties of those lower bounds.
However, in some cases, two entities that do not partition a collection in the same way may have the same value of a lower bound. For example, suppose entity $a$ partitions a collection of 16 sets into 9 and 7 sets, and entity $b$ partitions the same collection into 10 and 6 sets. With  
$\lceil \log_{2}9 \rceil = \lceil \log_{2}10 \rceil = 4$,
both entities will have the same value of the lower bound on height. When there are such ties, we select an entity that most evenly partitions the collection to differentiate between entities with the same value of cost lower bound.

Though these strategies seem different from each other, it can be shown that the existing strategies discussed above and our 1-step cost lower bound $LB_1$ select the same entity for the binary decision tree problem. Hence, they all achieve the same $(\ln n + 1)$-approximation factor. 

\begin{lemma}\label{lem:strategies-comp}
Given a collection $C$, the strategies (a) information gain, (b) indistinguishable pairs, and (c) 1-step cost lower bound, $LB_1$, select the same entity that partitions $C$ most evenly into two sub-collections.
\end{lemma}

\begin{proof}
(a) In (\ref{eqn:info-gain}), since $|C|$ is constant and $|C1| + |C2| = |C|$, the quantity $|C1| * \log_2|C1| + |C2| * \log_2|C2|$ is minimum when $C$ is most evenly partitioned into $C1$ and $C2$. Hence, the entity that partitions $C$ most evenly has the largest information gain and is selected by \textit{information gain} strategy.

\noindent
(b) Similarly, in (\ref{eqn:indg}), $|C1| * (|C1| - 1) + |C2| * (|C2| - 1)$ is minimum when $C$ is most evenly partitioned. Therefore, the entity that partitions $C$ most evenly has the minimum value of Indg() and is selected by \textit{indistinguishable pairs} strategy.

\noindent
(c) It can also easily be seen from (\ref{eqn:lb-1-e-ad}) and (\ref{eqn:lb-1-e-h}) (after replacing $LB\_AD_0$ and $LB\_H_0$ with their respective values from (\ref{eqn:lb-0-ad}) and (\ref{eqn:lb-0-h})) that, an entity that most evenly partitions the collection $C$ into $C1$ and $C2$, gives the minimum value of $LB_1$ in (\ref{eqn:lb-1}), thus is selected by our \textit{1-step cost lower bound} strategy.
\end{proof}

This paper builds on top of our cost lower bounds, and this offers a few benefits compared to other entity selection strategies. First, the cost functions are simple and intuitive, offering an easy choice between the average case and worse case costs. Second, we can develop efficient and effective $k$-steps lookahead strategies using $LB_k$ with $k > 1$. In particular, we develop an effective pruning strategy that significantly reduces the search space and the runtime of our lookahead strategies without affecting the cost.
Although $k$-steps lookahead strategies have been studied for entropy~\cite{Bonifati14interactiveinference} and information gain~\cite{Esmeir2004}, we are not aware of similar pruning strategies developed for these other measures.

\subsection{Pruning} \label{sec:pruning}
We want to find a decision tree that requires the least number of interactions with the user for a set discovery hence has the least cost.
However, exhaustive searching the space of possible trees for the one with the least cost is computationally intensive and not always feasible.
We propose a \textit{pruning} strategy that allows certain choices of entities for decision tree nodes to be safely rejected to reduce the size of the search space without affecting the correctness.

\begin{lemma}\label{lem:pruning-correctness}
Let $LB_{k}(C, e)$ denote our lower bound of cost for entity $e$ in collection $C$ by looking
k-steps ahead, and suppose entity selection is done based on $LB_{k}$, $k$-steps cost lower bound for some $k$. Consider entities $e_1$ and $e_2$, both in $C$. If $LB_{l}(C, e_2) \geq LB_{k}(C, e_1)$ for $l \leq k$, then $e_2$ can be pruned without affecting the correctness of the search.
\end{lemma}

\begin{proof}
Based on Lemmas~\ref{lem:least-cost-comp} and \ref{lem:least-cost-e-comp} $LB_{k}(C, e_2)$ cannot be smaller than $LB_{k}(C, e_1)$
when $LB_{l}(C, e_2) \geq LB_{k}(C, e_1)$ for $l \leq k$. Hence 
$e_2$ can be pruned without affecting the correctness of the search.
\end{proof}

As an example, consider the collection of sets shown in Fig.~\ref{fig:example-sets}, denoted as $C1$, and let H be our cost metric. The entities $c$ and $d$ are present in 3 sets and absent in 4 sets. Hence, the 1-step lower bound, $LB\_H_{1}()$, for entities $c$ and $d$ is $max(log_2(3), log_2(4))+1 = 3$. Similarly, 1-step lower bound for all other informative entities is 4. Suppose, we are using 3-steps cost lower bound for entity selection. The 3-steps lower bound for $d$, $LB\_H_{3}(C1, d)$, is 3. Since $LB\_H_{1}()$ for all other entities is not less than 3, any further calculation for them can be pruned safely.

Now, consider another collection where the sets are the same as in collection $C1$ except $S1 = \{a, b, c\}$ and $S4 = \{a, b, c, d, g, h\}$, and let us denote this collection with $C2$. The set counts for all entities are as before, hence the 1-step lower bound, $LB\_H_{1}()$, for the informative entities remain the same as in collection $C1$. But, the 3-steps lower bound for $d$, $LB\_H_{3}(C2, d)$, is 4 now. Therefore, we cannot prune the 3-steps lower bound calculation for entity $c$ using the 1-step lower bound, $LB\_H_{1}(C2, c)$, which is 3. Thus, we calculate the 2-steps lower bound for $c$, $LB\_H_{2}(C2, c)$, which is 4. Now, any further lower bound calculation for $c$ can be pruned using the 2-steps lower bound since it is not less than the already calculated least 3-steps lower bound for entity $d$.

\subsubsection{Implementation}
There are several places where our pruning is applied.
First, entities are sorted based on their 1-step lower bounds in non-decreasing order, the $k$-steps lower bounds for entities are calculated in that order, and the least value found so far is updated accordingly. If the 1-step lower bound of an entity $e$ is not less than the already found least $k$-steps lower bound, then the k-steps lower bound calculations of entity $e$ and all the subsequent entities in the sorted order are pruned.

Second, when calculating the $k$-steps lower bound for an entity, the already found least value is used to set an upper limit for each of the recursive steps of the calculation. Whenever the upper limit is reached, the rest of the $k$-steps lower bound calculation for the current entity is pruned. Since, for an entity $e$ to be selected,  $LB_{k}(C, e)$ needs to be less than the already found least value of the lower bound (AFLV), if $e$ partitions a collection $C$ into $C1$ and $C2$, the upper limit (UL) for the accepted value of $LB_{k-1}(C1)$ can be calculated using (\ref{eqn:lb-k-e-ad}), for the cost metric AD, by replacing $LB\_AD_{k}(C, e)$ with AFLV and $LB\_AD_{k-1}(C2)$ with the least possible value $LB\_AD_{0}(C2)$ as
\begin{equation}\label{eqn:upper-limit-1-ad}
  UL(C1) =  \frac{(AFLV - 1) * |C| - |C2| * LB\_AD_{0}(C2)}{|C1|},
\end{equation}
and similarly using (\ref{eqn:lb-k-e-h}), for the cost metric H, as
\begin{equation}\label{eqn:upper-limit-1-h}
  UL(C1) =  AFLV - 1.  
\end{equation}
Once the actual value of $LB_{k-1}(C1)$ is calculated, the upper limit (UL) for the accepted value of $LB_{k-1}(C2)$ can be calculated, for the cost metric AD, as

\begin{equation}\label{eqn:upper-limit-2-ad}
  UL(C2) =  \frac{(AFLV - 1) * |C| - |C1| * LB\_AD_{k-1}(C1)}{|C2|},
\end{equation}
and for the cost metric H, as
\begin{equation}\label{eqn:upper-limit-2-h}
  UL(C2) =  AFLV - 1.  
\end{equation}

\subsection{Lookahead Strategies}\label{sec:lookaheads}
Now that we have covered our cost functions and pruning strategies, we present our $k$-steps lookahead strategy and two variations of it, for selecting the next question to ask by looking $k$-steps ahead. These strategies choose an entity based on the $k$-steps lower bound for the cost of a collection, as discussed in Section~\ref{sec:entity-selection}. When there are ties between two or more entities in terms of cost, the entity that partitions the collection most evenly is chosen. If there are still ties for the choice of entities, then an entity is selected randomly from the set of candidates. The algorithm applies our pruning strategy in every step where the search space can be cut without compromising the required number of questions, as discussed in Section~\ref{sec:pruning}.

\subsubsection{$k$-Lookahead with Pruning ($k$-LP)}
Algorithm~\ref{alg:k-lookahead} presents our \textit{$k$-lookahead with pruning} strategy. It takes a collection $C$ of unique sets, the number of steps $k$ to look ahead, and an upper limit $ul$ of the $k$-steps lower bound for an entity to be selected, as input. Initially, the upper limit is set to a large number. Then, it sorts the entities in the collection based on their partitioning capability from the most even to the least even (Line 11). Since, the entity that partitions a collection most evenly has the minimum value of the 1-step cost lower bound, the entities will also be sorted based on their 1-step lower bound of cost in non-decreasing order. This way, an entity with both the minimum lower bound and also the most even partitioning capability is considered first, breaking possible ties on the cost lower bound. For each entity in the sorted order that partitions the collection $C$ into two sub-collections $C^+$ and $C^-$, the $(k-1)$-steps lower bounds for $C^+$ and $C^-$ are calculated by recursively calling the algorithm (Lines 16-32). Those quantities are plugged into (\ref{eqn:lb-k-e-ad}) or (\ref{eqn:lb-k-e-h}) (depending on the cost metric used) to obtain the $k$-steps lower bound $l$ for each entity (Line 33). The algorithm keeps track of the entity with the least $k$-steps lower bound $l$ and sets it as the upper limit $ul$ for the next entity to be considered (Lines 33-35). Since the entities are sorted, if it finds an entity with an equal or larger 1-step lower bound than the upper limit $ul$, the algorithm stops early and prunes all the remaining entities (Lines 14-15). To further reduce the search space, it calculates the upper limits for $C^+$ (using (\ref{eqn:upper-limit-1-ad}) or (\ref{eqn:upper-limit-1-h})) and $C^-$ (using (\ref{eqn:upper-limit-2-ad}) or (\ref{eqn:upper-limit-2-h})) and passes them to the recursive call (Lines 22-23 and 29-30). If no entity can be selected with a lower value of $(k-1)$-steps lower bound than the calculated upper limit, then it stops processing the current entity and moves to the next entity (Lines 24-25 and 31-32). Finally, the algorithm returns an entity $e$ with the minimum $k$-steps lower bound of cost (Lines 7-10 or 38). To speed up the calculations, the algorithm uses memoization by storing and reusing the results for different inputs of collection $C$ and steps, $k$ (Lines 1-6, 9, and 37).

\begin{algorithm}[tb]
\caption{K-Lookahead with Pruning (K-LP)}
\label{alg:k-lookahead}
\begin{flushleft}
\textbf{Input:} collection $C$ of unique sets, steps $k$, and upper limit $ul$ of the $k$-steps cost lower bound for an entity to be selected \\
\textbf{Output:} selected entity and it's $k$-steps cost lower bound\\
\end{flushleft}
\begin{algorithmic}[1]
\If {$(C, k) \in Cache$}
    \State $e$, $l \leftarrow Cache[(C, k)]$
    \If {$ul \le l$}
        \State \textbf{return} $null$, $l$
    \ElsIf {$e \neq null$}
        \State \textbf{return} $e$, $l$
    \EndIf
\EndIf
\If {$k = 1$}
    \State \textbf{let} entity $e$ to most evenly partition $C$
    \State $Cache[(C, k)] \leftarrow (e$, $LB_{1}(C, e))$
    \State \textbf{return} $Cache[(C, k)]$
\EndIf
\State $SE \leftarrow$ sort entities according to most even partitioning of $C$
\State $e \leftarrow null$
\For {each entity $e_i \in SE$}
    \If {$LB_{1}(C, e_i) \geq ul$}
        \State \textbf{break}
    \EndIf
    \State \textbf{let} $C^+$ be the collection $\{S_j \in C \mid e_i \in S_j\}$
    \State $C^- \leftarrow C - C^+$
    \If {$|C^+| = 1$}
        \State $l^+ \leftarrow 0$
    \Else
        \State $lb^- \leftarrow LB_{0}(C^-)$
        \State $ul^+ \leftarrow$ \textbf{Upper-Limit} $(ul, |C^+|, lb^-, |C^-|, |C|)$
        \State $e^+, l^+ \leftarrow$ \textbf{K-LP} $(C^+, k - 1, ul^+)$
        \If {$e^+ = null$}
            \State \textbf{continue}
        \EndIf
    \EndIf
    \If {$|C^-| = 1$}
        \State $l^- \leftarrow 0$
    \Else
        \State $ul^- \leftarrow$ \textbf{Upper-Limit} $(ul, |C^-|, l^+, |C^+|, |C|)$
        \State $e^-, l^- \leftarrow$ \textbf{K-LP} $(C^-, k - 1, ul^-)$
        \If {$e^- = null$}
            \State \textbf{continue}
        \EndIf
    \EndIf
    \State $l \leftarrow$ \textbf{K-Steps-Lower-Bound} $(|C^+|, l^+, |C^-|,  l^-, |C|)$
    \If {$l < ul$}
        \State $ul \leftarrow l$
        \State $e \leftarrow e_i$
    \EndIf
\EndFor
\State $Cache[(C, k)] \leftarrow (e$, $ul)$
\State \textbf{return} $e$, $ul$ 
\end{algorithmic}
\end{algorithm}

Two important observations can be made about our k-LP algorithm. First, it can be shown that the algorithm finds \textit{an optimal solution} if $k$ is set to the height of an optimal tree or a greater value. Second, the early
stopping opportunities, which are based on our pruning strategy presented earlier, sets apart our lookahead strategies from the existing lookaheads in literature~\cite{Esmeir2004,Bonifati14interactiveinference}.

For a collection of $n$ unique sets and $m$ distinct entities, the runtime of Algorithm~\ref{alg:k-lookahead} is $O(m^{k}n)$ since finding 1-step lower bounds for $m$ entities is $O(mn)$ and in each recursive step, there will be $O(m)$ calls to the next step. Our next two strategies further reduce the time by setting bounds on the number of candidate entities.

\subsubsection{$k$-LP with Limited Entities ($k$-LPLE)}
Despite all the pruning done using our lower bounds, the runtime of our $k$-LP algorithm increases as a polynomial function of $m$ (e.g., quadratic for $k=2$), and the algorithm becomes very inefficient for large values of $k$. On the other hand, the chance of constructing a better tree increases as we increase $k$. 
\revised{One good trade-off is to limit the number of candidate entities in each step of the lower bound calculation to $q<m$, ranked in terms of the 1-step lower bound of cost. For $q<<m$, the reduction in runtime can be significant, analogous to setting a \textit{beam size} in deep learning algorithms~\cite{zhang2021dive}.}
This can be implemented by adding an extra input $q$ to Algorithm~\ref{alg:k-lookahead}, modifying Line 11 so that $SE$ contains only the first $q$ sorted entities, and passing $q$ on the recursive calls to the algorithm in Lines 23 and 30.

\subsubsection{$k$-LP with Limited but Variable number of Entities($k$-LPLVE)}
The runtime of k-LPLE may further be reduced by greedily considering only a single entity in each recursive step of the k-steps lower bound calculation for an entity. The intuition here is that an entity with the smallest 1-step lower bound is more probable to be the best choice. Hence, our k-LPLVE strategy limits the number of candidate entities to only one (with the least 1-step lower bound) during each step of the lower bound calculation.
With this strategy, the search time is expected to reduce further, but the quality of the results is not expected to change much (see Section~\ref{sec:experiments} for evaluation results).
This strategy can be implemented by performing the same modifications as k-LPLE in Algorithm~\ref{alg:k-lookahead} except that when the function is called from outside with $q$, $SE$ in Line 11 takes the first $q$ sorted entities during that call and only the first entity during the subsequent recursive calls to the function.

\subsection{Set Discovery}\label{sec:set-discovery}
The set discovery scheme studied in this paper is an interactive process that starts with an initial question posed to the user and continues with follow-up questions based on the user's answers. The lookahead strategies
choose an entity to be the next question, which is expected to minimize the cost of discovering the user's desired set. With each user feedback, the same selection process continues until the user's desired set is discovered or the user is satisfied with the refined sub-collection of sets and does not want to answer more questions.

The general approach for \textit{set discovery} is presented in Algorithm~\ref{alg:set-discovery}. It takes the entire collection $C$ of unique sets and a user-provided initial set $I$ as inputs and finds the sub-collection $CS$ containing all the supersets of $I$ in $C$ (Lines 2-4). It then iteratively, selects the best entity $e$ according to the entity selection strategy denoted by $\mathbf{\Upsilon}$, asks the user a question about the presence of that entity in the desired set, and re-calculates the sub-collection $CS$ of candidate sets based on the user feedback until a single set is left or the halt condition $\mathbf{\Gamma}$ (e.g., the user does not want to answer more questions) is met (Lines 5-12). Finally, it returns the remaining sets that are consistent with the user's answers (Line 13).
The runtime of Algorithm~\ref{alg:set-discovery} depends on the number of questions required to discover the desired set and the strategy used. In the worst case, the number of questions can be $n - 1$ for a collection of $n$ sets.

\begin{algorithm}[!tbp]
\caption{Set Discovery}
\label{alg:set-discovery}
\begin{flushleft}
\textbf{Inputs:} collection $C$ of unique sets and initial set $I$\\
\textbf{Output:} sets that are consistent with the user's answers\\
\textbf{Parameter:} entity selection strategy $\mathbf{\Upsilon}$ and halt condition $\Gamma$
\end{flushleft}
\begin{algorithmic}[1]
\State $CS \leftarrow  \emptyset $
\For {each set $S_i \in C$}
    \If {$I \subseteq S_i$}
       \State $CS \leftarrow CS \cup \{S_i\}$
    \EndIf
\EndFor
\While {$|CS| > 1$ and $\Gamma$ is $false$}
    \State $e \leftarrow \mathbf{\Upsilon} (CS)$
    \State $\alpha \leftarrow$ \textbf{query} user about the presence of $e$ in target set
    \State \textbf{let} $P$ be the collection $\{S_i \in CS \mid e \in S_i\}$
    \If {$\alpha$ is $true$}
        \State $CS \leftarrow P$
    \Else
        \State $CS \leftarrow CS - P$
    \EndIf
\EndWhile
\State \textbf{return} $CS$ 
\end{algorithmic}
\end{algorithm}

\noindent\textbf{Offline tree construction}
Our tree construction may be done offline for static collections, for example, when the initial query sets are known in advance or are always empty. 
\revised{An offline construction may be useful when the same decision tree is constructed multiple times or is used by multiple queries.}
Algorithm~\ref{alg:tree-construction} provides the steps for precomputing a decision tree on a collection of sets. With the decision tree constructed offline, a set discovery can be efficiently performed by asking questions and following only a single path through the tree in real-time.

Algorithm~\ref{alg:tree-construction} takes a collection $C$ of unique sets as input. If the collection has only one set, then it constructs a tree $T$ consisting of a single node with the only set $G$ (Lines 1-3). Otherwise, the algorithm selects the best entity $e$ using the entity selection strategy denoted by $\mathbf{\Upsilon}$ (Line 5). It recursively constructs the subtrees $T^+$ and $T^-$ for the two sub-collections $C+$ and $C-$ respectively (Lines 6-9). Finally, a tree $T$, consisting of a root node $e$ and two child subtrees ($T^+, T^-$), is constructed and returned (Lines 10-11).

There are $n - 1$ internal nodes in a full binary decision tree representing a collection of $n$ sets and $m$ entities, and each internal node requires a k-steps lookahead, which costs $O(m^{k}n)$. Hence the runtime of  
Algorithm~\ref{alg:tree-construction} is $O(m^{k}n^{2})$.

\begin{algorithm}[!tbp]
\caption{Tree Construction}
\label{alg:tree-construction}
\begin{flushleft}
\textbf{Input:} collection $C$ of unique sets \\
\textbf{Output:} a decision tree representation of the input collection\\
\textbf{Parameters:} entity selection strategy $\mathbf{\Upsilon}$
\end{flushleft}
\begin{algorithmic}[1]
\If {$|C| = 1$} 
    \State \textbf{let} $G$ be the only element of $C$
    \State $T \leftarrow$ \textbf{Tree} $(G, null, null)$
\Else
    \State $e \leftarrow \mathbf{\Upsilon} (C)$
    \State \textbf{let} $CS^+$ be the collection $\{S_i \in C \mid e \in S_i\}$
    \State $CS^- \leftarrow C - CS^+$
    \State $T^+ \leftarrow$ \textbf{Tree-Construction} $(CS^+)$
    \State $T^- \leftarrow$ \textbf{Tree-Construction} $(CS^-)$
    \State $T \leftarrow$ \textbf{Tree} $(e, T^+, T^-)$
\EndIf
\State \textbf{return} $T$ 
\end{algorithmic}
\end{algorithm}


\section{Experiments}
\label{sec:experiments}
 This section reports an experimental evaluation of our algorithms and pruning strategies on both real and synthetic data and under different parameter settings.

\subsection{Evaluation Setup}
As our evaluation measures, we
study (a) the effectiveness of our algorithms in finding a ``good" solution for the problem of set discovery, (b) the effectiveness of our pruning strategy in reducing the size of the search space, (c) the efficiency of our algorithms in terms of the running time, and (d) the scalability of our algorithms with both the number and the size of sets. Our results are compared to the relevant algorithms in the literature (when applicable).

The effectiveness of our set discovery is measured in terms of the \textit{number of questions} to be answered by a user looking for a target set. Without knowing much about the target set of a user, we assume all sets that contain an initially provided set are equally likely. With this, the effectiveness may be defined in terms of the \textit{average number of questions} or the \textit{maximum number of questions} to be answered by a user. These quantities also represent the average depth of the leafs (AD) and the height (H) of a decision tree that is constructed. 

The efficiency of an entity selection algorithm is measured in terms of the \textit{tree construction time}, which is the time needed to construct a decision tree using the selection strategy. It can be noted that the tree construction time is different from the time spent when searching for a specific set (\textit{discovery time}). For the former, Algorithm~\ref{alg:tree-construction} constructs a whole tree with all sets placed at the leaves and the internal nodes giving the paths to all sets at the leaves, whereas for the latter, Algorithm~\ref{alg:set-discovery} only constructs a path from the root to the target set.
The latter is much less if the wait time for user responses is excluded.

The algorithms being evaluated include entity selection using $k$-LP, $k$-LPLE, and $k$-LPLVE strategies. For a comparison with entity selection strategies from the literature, our evaluation also includes \textit{information gain} (InfoGain)~\cite{Quinlan1986} and \textit{gain-$k$}~\cite{Esmeir2004}. Our reported result for  \textit{information gain} holds for \textit{indistinguishable pairs}~\cite{Roy2008} and 1-step lookahead, i.e. gain-$k$ and $k$-LP with $k = 1$, since they all select the same entity, as shown earlier  (Lemma~\ref{lem:strategies-comp}).

Our algorithms were implemented in Python 3 and our experiments were run on a 64-bit 
machine with Intel(R) Core i5-9300H @2.40 GHz processor and 8 GB RAM.

\subsection{Datasets and Queries}
We conduct our experiments on two datasets for set discovery, including \textit{web tables}, which consists of a collection of entity sets extracted from the columns of various web tables, and \textit{synthetic datasets}, where large collections of sets are generated following some distributions. The former evaluates our algorithms on a real dataset, whereas the latter assesses the scalability of our strategies under different collection sizes and parameter settings.

We also evaluate our algorithms on the task of query discovery, based on a \textit{baseball} database.

\subsubsection{Web tables}
Our web table dataset is collected from a 2014 snapshot of Wikipedia. We extract tables in document text and treat each of their columns as a set based on the observation that each column has a domain and the values are drawn from that domain. As an example, one set includes 58 NBA players including Steve Nash, Kobe Bryant, and Tracy McGrady. The sets are diverse, covering many domains of interest, but also noisy. We remove any set that has less than three distinct elements and sets that consist of all numbers. We further remove duplicate entries, making each list a pure set, and a few frequent keywords such as unknown, tba, total. After those cleanings, we obtain 1,407,178 unique sets containing in total 6,312,409 distinct entities. Examples of sets and queries are given elsewhere \cite{Hasnat2021Thesis}
For our experiments, we considered each combination of two entities as a possible initial example set and the sets that contained the two example entities as candidates. This resulted in 14,491 initial sets, giving us the same number of sub-collections with at least 100 sets in each. 
The choice of two example entities was based on the observation that at least two entities from a semantic class are required to unambiguously represent the class. As an example, Liverpool may represent both a ``City" and a ``Football Club" whereas Liverpool and Arsenal together do not represent the ``Cities" semantic class. The number of sets in the selected collections was in the range of $[100, 11219]$ with an average of 390 and a standard deviation of 478, and the number of distinct entities was in the range of $[15, 15186]$ with an average of 3,112 and a standard deviation of 2,379. We constructed decision trees for all the selected sub-collections to evaluate our strategies.

\subsubsection{Synthetic data}
To study the performance of our entity selection strategies under different data distributions as well as the scalability with the number of entities and sets in the collection, we generated a few synthetic set collections.
The set generation follows a copy-add preferential mechanism where some elements are copied from an existing set and the rest of the elements are added from a universe of elements. Similar copying models are used in other domains (e.g., the dynamics of the web graph~\cite{kumar2000stochastic}, the copying and publishing relationships between data sources~\cite{dong2009truth}, etc.).
Each set has two parameters: a set size $s$, chosen randomly from a range of values (e.g., $[50, 100]$), and an overlap ratio $\alpha \in [0,1)$. For each set, we choose a size $s$ from the range of possible sizes randomly and an overlap ratio $\alpha$. Then $\alpha * s$ elements are copied from a previously generated set and $(1-\alpha) * s$ elements are added from the entity universe. If a previously generated set does not exist or does not have enough elements, then additional elements are selected from the entity universe to bump up the set size to $s$.
We generated 19 synthetic collections by varying the overlap ratio $\alpha$, the range of set sizes $d$, and the
 number of sets $n$.
Table~\ref{tab:synthetic-datasets} gives some information about these collections including the number of distinct entities in each collection. For this dataset, no entities were selected as query entities (i.e., the user-provided initial set is considered empty), and all the sets in each collection are considered as possible target sets.

\begin{table}[!tbp]
    \resizebox{\columnwidth}{!}{%
    \begin{subtable}[b]{0.33\linewidth}
        \centering
        \begin{tabular}{|c|c|c|c|}
        \hline
        \thead{Overlap \\ ratio \\$\alpha$} & \thead{Number \\of \\distinct \\entities}\\
        \hline
        0.99 & 23k\\
        0.95 & 36k\\
        0.90 & 59k\\
        0.85 & 83k\\
        0.80 & 108k\\
        0.75 & 132k\\
        0.70 & 156k\\
        0.65 & 178k\\
        \hline
        \end{tabular}
        \vspace*{2mm}
        \caption{$n = 10$k, $d = 50 - 60$}
        \label{tab:synthetic-alpha}
    \end{subtable}%
    \begin{subtable}[b]{0.33\linewidth}
        \centering
        \begin{tabular}{|c|c|c|c|}
        \hline
        \thead{Number \\of \\sets $n$} & \thead{Number \\of \\distinct \\entities}\\
        \hline
        10k & 59k\\
        20k & 125k\\
        40k & 216k\\
        80k & 385k\\
        160k & 622k\\
        \hline
        \end{tabular}
        \vspace*{2mm}
        \caption{$\alpha = 0.9$, $d = 50 - 60$}
        \label{tab:synthetic-n}
    \end{subtable}
    \begin{subtable}[b]{0.33\linewidth}
        \centering
        \begin{tabular}{|c|c|c|c|}
        \hline
        \thead{Set size \\range \\ $d$} & \thead{Number \\of \\distinct \\entities}\\
        \hline
        50-100 & 119k\\
        100-150 & 150k\\
        150-200 & 180k\\
        200-250 & 214k\\
        250-300 & 249k\\
        300-350 & 283k\\
        \hline
        \end{tabular}
        \vspace*{2mm}
        \caption{$n = 10$k, $\alpha = 0.9$}
        \label{tab:synthetic-d}
    \end{subtable}%
    }
    \caption{Synthetic data by varying (a) overlap ratio $\alpha$, (b) number of sets $n$ and (c) set size range $d$}
    \label{tab:synthetic-datasets}
\end{table}

\subsubsection{Baseball database}
The baseball database~\cite{baseball} is a complex, multi-relation database that contains batting, pitching, and fielding statistics plus standings, team stats, player information, and more for Major League Baseball (MLB) covering the years between 1871 and 2020. Our experiment is based on the \textit{People} table which contains information about name, birth, death, height, weight, batting and throwing hand, etc., of 20,185 baseball players. For our experiment, we considered only CNF (conjunctive normal form) queries with conditions on columns \textit{birthCountry}, \textit{birthState}, \textit{birthCity}, \textit{birthYear}, \textit{birthMonth}, \textit{birthDay}, \textit{height}, \textit{weight}, \textit{bats}, and \textit{throws} of the \textit{People} table. At first, we constructed 7 target queries that could be interesting to a user. Table~\ref{tab:target-queries-baseball} describes the target queries and the number of tuples in their outputs. Then, for each target query, we randomly selected 2 output tuples as the example tuples and generated candidate CNF queries that contain the example tuples in their output. The candidate queries are generated using the following simple steps:

\begin{enumerate}
    \item The columns are grouped into categorical and numerical with columns \textit{birthCountry}, \textit{birthState}, \textit{birthCity}, \textit{birthMonth}, \textit{birthDay}, \textit{bats}, and \textit{throws} treated as categorical and \textit{birthYear}, \textit{height}, and \textit{weight} treated as numerical in our experiments.
    
    \item A few reference values are defined for each numerical column. For examples, \textit{height}: \{60, 65, 70, 75, 80\}, \textit{weight}: \{120, 140, 160, 180, 200, 220, 240, 260, 280, 300\}, and \textit{birthYear}: \{1850, 1870, 1890, 1910, 1930, 1950, 1970, 1990\}.
    
    \item A selection condition on each categorical column is constructed as the disjunctions of the unique values of the example tuples for that column. For example, if the birth city of an example player is Chicago and that of another player is Seattle, then the selection condition is $birthCity = ``Chicago" \vee birthCity = ``Seattle"$, whereas if the birth city of all example players is Chicago, then the selection condition is $birthCity = ``Chicago"$.
    
    \item A few selection conditions on each numerical column are constructed using the possible intervals of the reference values that contain the values of all example tuples. For example, if the height of an example player is 62 and that of another player is 73, then the possible selection conditions on \textit{height} are $height > 60 \wedge height < 75$, $height > 60 \wedge height < 80$, $height > 60$, $height < 75$, and $height < 80$.
    
    \item Each selection condition on a column yields a candidate query, and the conjunction of any two selections on different columns provide additional candidate queries. \sloppy For example, $\sigma_{birthCity = ``Los Angeles"} (People)$ is a query with selection condition on a single column, whereas $\sigma_{birthCity = ``Los Angeles" \wedge height > 70 \wedge height < 80} (People)$ is a query with selection conditions on two columns. Similarly, candidate queries with selection conditions on more columns can be generated. Our experiments consider queries with selection conditions on up to two columns.
\end{enumerate}

\begin{table}[!tbp]
    \centering
    \resizebox{\columnwidth}{!}{%
    \begin{tabular}{|c|c|c|}
    \hline
    \thead{Target \\query} & \thead{Query \\description} &  \thead{Number of \\output tuples}\\
    \hline
    
    T1 & $\sigma_{birthCountry = ``USA" \wedge birthYear > 1990} (People)$ & 892\\
    T2 & $\sigma_{birthCity = ``Los Angeles" \wedge height > 70 \wedge height < 80} (People)$ & 201\\
    T3 & $\sigma_{bats = ``L" \wedge throws = ``R"} (People)$ & 2179\\
    T4 & $\sigma_{birthCountry = ``USA" \wedge bats = ``B"} (People)$ & 939\\
    T5 & $\sigma_{birthMonth = 12 \wedge birthDay = 25} (People)$ & 65\\
    T6 & $\sigma_{height > 75 \wedge weight > 260} (People)$ & 49\\
    T7 & $\sigma_{height < 65 \wedge weight < 160} (People)$ & 26\\
    \hline
    \end{tabular}%
    }
    \vspace*{0.5mm}
    \caption{Target queries for the baseball database}
    \label{tab:target-queries-baseball}
\end{table}

Once the candidate queries were generated, we applied our set discovery strategy to discover the target query. The user answers about the membership of the presented tuples were simulated by verifying them against the output of the target query. Table~\ref{tab:example-tuples-baseball} provides information about the selected example tuples for each target query, the number of generated candidate queries from the example tuples, and the average number of tuples in the output of those candidate queries.

\begin{table}[!tbp]
    \centering
    \begin{tabular}{|c|c|c|c|}
    \hline
    \thead{Target \\query} & \thead{Player ids of \\example tuples} & \thead{\# of candidate\\ queries} & \thead{Average number \\of output tuples}\\
    \hline
    T1 & baragca01, phillev01 & 776 & 9404.24 \\
    T2 & ryanbr01, edwarda01 & 987 & 11254.35\\
    T3 & ellioal01, drumrke01 & 940 & 10612.07\\
    T4 & dashnle01, craigro02 & 916 & 10957.30\\
    T5 & brownll01, ellerfr01 & 1339 & 9772.70 \\
    T6 & evansde01, fulchje01 & 600 & 7187.00\\
    T7 & emmerbo01, gearidi01 & 1189 & 7795.78\\
    
    \hline
    \end{tabular}
    \vspace*{0.5mm}
    \caption{Information about selected example tuples and generated candidate queries on baseball database}
    \label{tab:example-tuples-baseball}
\end{table}

\subsection{Evaluation Results}\label{sec:exp-set-discovery}
\subsubsection{Choosing the parameters $k$ and $q$}
\revised{
As $k$ increases, the generated trees are expected to be closer to an optimal tree and when $k$ is set to the height of an optimal tree or greater, our algorithm finds an optimal tree. However, as $k$ increases, the running time increases dramatically. Parameter q acts similar to \textit{beam size} in deep learning, and setting this parameter allows us to increase $k$ without too much affecting the running time.}
To set the parameters $k$ and $q$ for our algorithms, we did run some experiments on our web tables dataset.
Fig.~\ref{fig:k-selection} shows that the runtime of $k$-LP increases by one to two orders of magnitude when the number of lookahead steps $k$ is increased from 2 to 3. At the same time, the average number of questions usually becomes less with higher $k$. To balance the runtime with the quality of the trees that are constructed, we set $k = 2$ for our experiments with the $k$-LP strategy. The runtime may also be kept low, while increasing $k$, using the $k$-LPLE strategy,
which limits the number of entities in each step. For our experiments with $k$-LPLE and $k$-LPLVE strategies, we set $k = 3$ and experiment with different values (up to 50) of the number of entities $q$. The average number of questions that are required remains almost the same when the value of $q$ exceeds 10, but the runtime increases significantly. Therefore, we set $q = 10$ for the $k$-LPLE and $k$-LPLVE strategies. The average numbers of questions for larger values of $q$ are almost the same hence are not reported here.

\begin{figure}[!htb]
    \centering
    \resizebox{.79\linewidth}{!}{%
        \includegraphics[width=\linewidth]{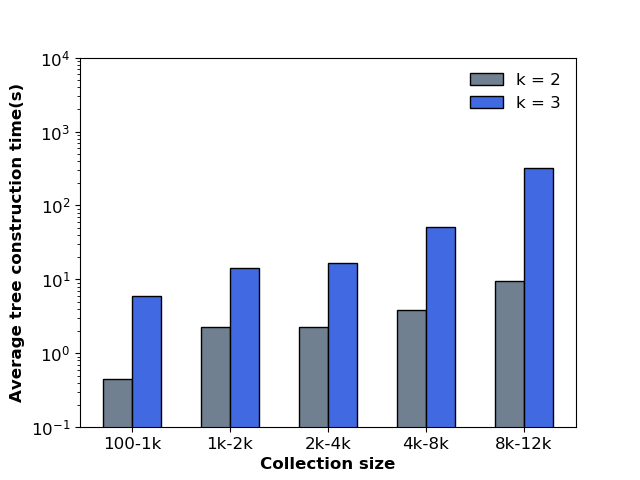}%
     }
    \caption{Tree construction time (seconds) for $k$-LP varying $k$ on web tables dataset}
    \label{fig:k-selection}
\end{figure}

\subsubsection{Comparison to  strategies in the literature}
\revised{
A strong baseline for comparison is \textit{information gain}~\cite{Quinlan1993}, which is also equivalent to \textit{indistinguishable pairs}, \textit{gain-$k$} with $k = 1$, and our $k$-LP with $k = 1$; they all select the same entity as discussed in Section~\ref{sec:entity-selection}.
Our evaluation shows improvements over InfoGain in the average number of questions with the cost metric AD and the maximum number of questions with the cost metric H.
}
The mean improvement in the maximum number of questions (H) is close to one, whereas the mean improvement for the average number of questions (AD) is less due to the facts that the improvement is averaged over all sets in each sub-collection and that the average number of questions for InfoGain is already very close to the optimal (the average difference in the average number of questions with optimal solution for InfoGain is only about $0.048$) with little room for improvement. 
\revised{
The improvements in the number of questions over Info-Gain for all our reported methods (k-LP with k=2 and k-LPLE and k-LPLVE with k=3 and q=10) under both AD and H are all statistically significant at $\alpha=0.01$ using one-tailed t-test.
It should be noted that Info-Gain is a pretty strong baseline, and any
improvement in the number of questions is important. For example, if the questions are medical tests required to identify a disease, then a small reduction even in the average number of tests could save the patients a large amount of money and time to complete the tests.
}

    

\begin{table}[]
    \centering {
    \begin{tabular}{l|l l l l l l l}
             & T1 & T2 & T3 & T4 & T5 & T6 & T7 \\ \hline
         Avg & 97.3\% & 99.4\% & 99.1\% & 99.7\% & 88.5\% & 99.7\% & 99.9\% \\
         Min & 90.1\% & 94.6\% & 96.5\% & 98.0\% & 30.6\% & 98.1\% & 99.5\%
    \end{tabular}}
    \caption{Average and minimum number of entities pruned at all nodes for the Baseball dataset}
    \label{tab:ent-prubed-baseball}
\end{table}

\subsubsection{Effectiveness of our pruning}
The pruning proposed in this paper makes a huge difference in the tree construction time of all our strategies. 
\revised{On our Web tables dataset, more than 99\% of candidate entities are pruned at the root level (for both $k$-LP with $k=2$ and $k=3$), meaning no tree is constructed for those entities as the root. The pruning also happens at the nodes under the root. Table~\ref{tab:ent-prubed-baseball} shows the average and the minimum number of entities pruned at each node for the Baseball dataset at $k=2$. The results are almost the same for $k=3$.}
\revised{In most cases, more than 90\% of the entities are pruned, demonstrating the effectiveness of our lower bounds and the choice of questions.}
Fig. \ref{fig:pruning-k} shows the speedup on the web tables dataset, and Fig.~\ref{fig:pruning-synthetic} shows the same on the synthetic datasets.
The average speedup in runtime on the web tables dataset is in the range of two to three orders of magnitude for $k = 2$ and up to five orders of magnitude when $k = 3$. Since the runtime of gain-$k$ increases polynomially with the number of entities and exponentially with $k$, the speedup is more for larger values of $k$ and on datasets with a large number of entities and sets. This can be seen in 
Fig.~\ref{fig:pruning-k} for the web tables dataset where $k$ is varied from 2 to 3  and in
Fig.~\ref{fig:pruning-synthetic} for the synthetic dataset with a fixed $k$ and varying the number of sets.

\begin{figure*}[!htb]
    \begin{subfigure}[b]{0.49\linewidth}
        \centering
        \resizebox{.79\linewidth}{!}{%
            \includegraphics[width=\linewidth]{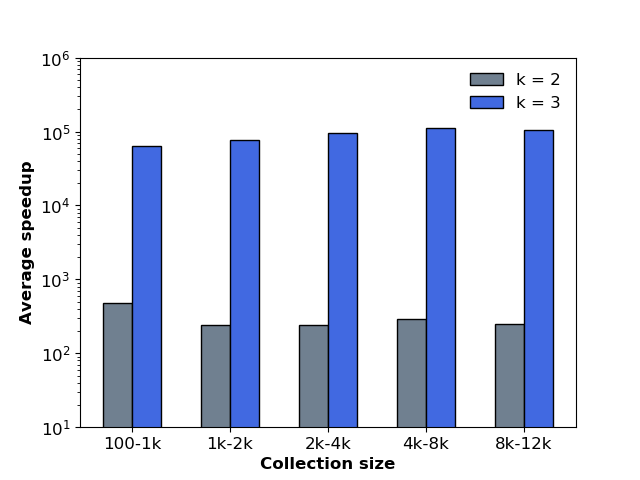}%
        }
       \caption{$k$-LP vs Gain-$k$ on web tables data\\ \hfill}
        \label{fig:pruning-k}
    \end{subfigure}
    \begin{subfigure}[b]{0.49\linewidth}
        \centering
        \resizebox{.79\linewidth}{!}{%
            \includegraphics[width=\linewidth]{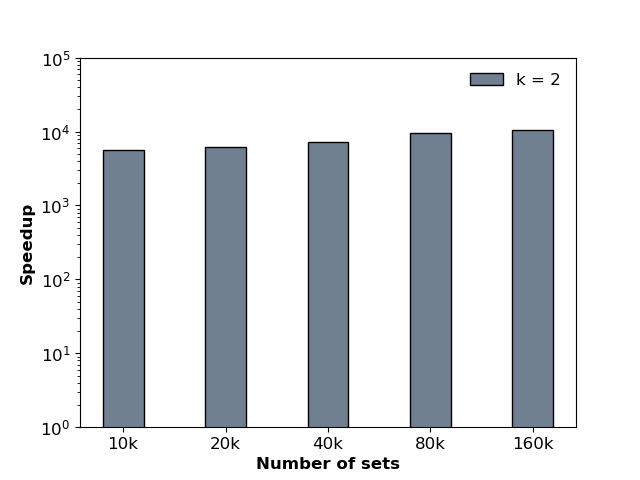}%
        }
       \caption{$k$-LP vs Gain-$k$ on synthetic data\\ \hfill}
        \label{fig:pruning-synthetic}
    \end{subfigure}
    \caption{Speedup of our strategies because of pruning}
    \label{fig:speedup-pruning}
\end{figure*}

\subsubsection{Performance varying the overlap between sets}
One factor that affects the performance of a set discovery is the amount of overlap between sets. Consider an extreme case where there is no overlap between sets. With $n$ sets, one needs to ask roughly $n/2$ questions on average ($n-1$ questions in the worst case) to find a target set. As the overlap between sets increases, there is more chance to filter more than one set with each question.
To better understand this relationship between the overlap and the search performance, we varied the overlap ratio as in Table~\ref{tab:synthetic-alpha} for our synthetic dataset and measured the number of questions that were needed to discover each set.  
Fig.~\ref{fig:effects-overlapping} shows the average number of questions that were needed as the overlap ratio varied from 0.65 to 0.99. As the overlap ratio increases, both the average number of questions and the tree construction time decrease. When the overlap ratio becomes less than 0.90, the average number of questions starts showing an upward trend. This upward trend is expected to continue to the point where one needs to ask roughly $n/2$ questions on average ($n-1$ questions in the worst case) to find a target set. This happens, for example, when
all sets have the same elements except at least one more element that distinguishes each set from the rest. 

\begin{figure*}[tb]
\centering
\begin{minipage}{.3\textwidth}
  \centering\includegraphics[width=2.35in]{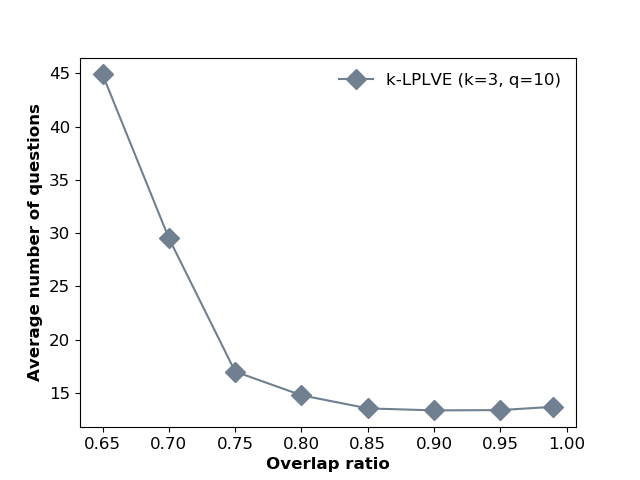}
\end{minipage}
\hspace*{0.2cm}
\begin{minipage}{.3\textwidth}
\centering\includegraphics[width=2.35in]{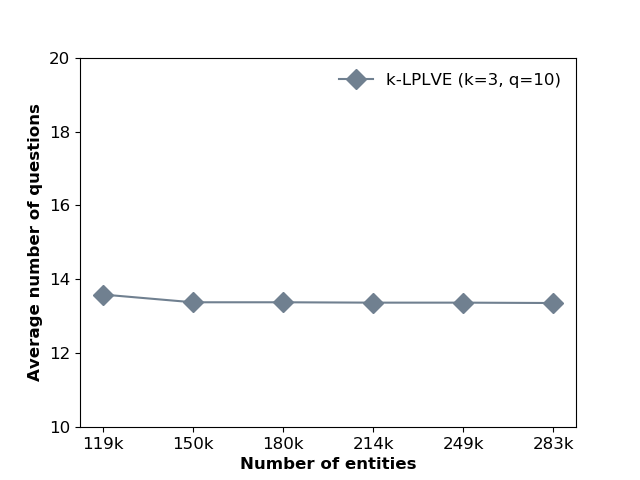}
\end{minipage}
\hspace*{0.2cm}
\begin{minipage}{.3\textwidth}
\centering\includegraphics[width=2.35in]{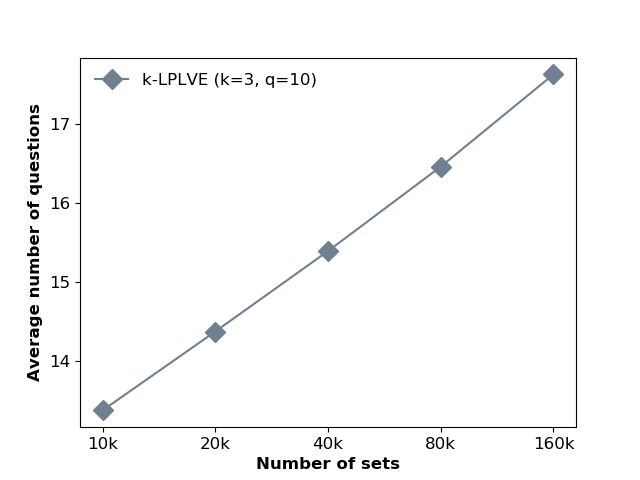}
\end{minipage}
\begin{minipage}{.3\textwidth}
  \centering\includegraphics[width=2.35in]{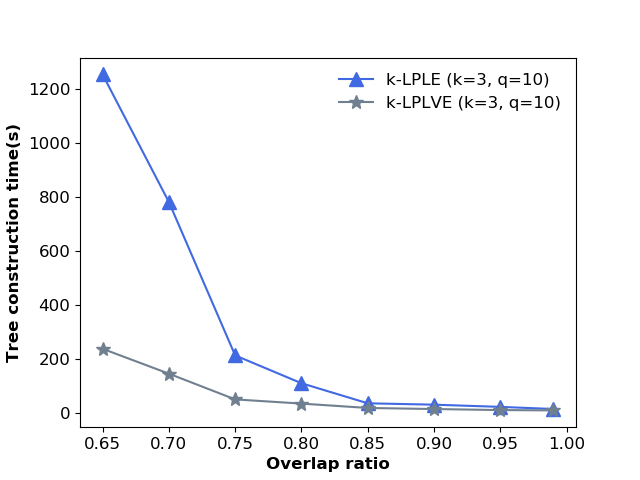}
  \caption{Effects of set overlaps on average number of questions (top) and tree construction time in seconds (bottom)}
  \label{fig:effects-overlapping}
\end{minipage}
\hspace*{0.2cm}
\begin{minipage}{.3\textwidth}
  \centering\includegraphics[width=2.35in]{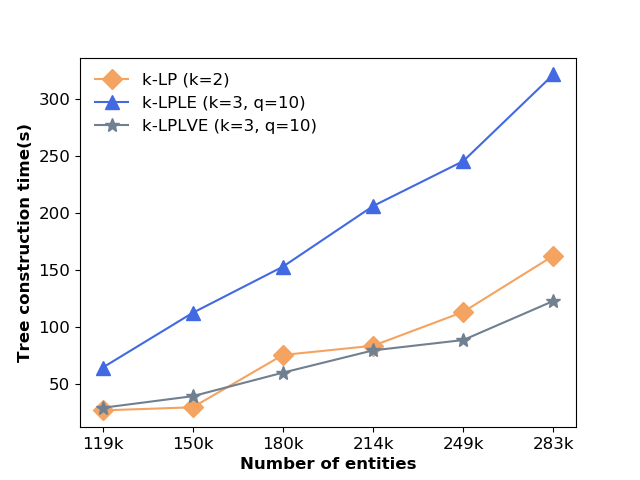}
  \caption{Effects of increasing the number of distinct entities in a collection on average number of questions (top) and tree construction time in seconds (bottom)}
  \label{fig:effects-set-size}
\end{minipage}
\hspace*{0.2cm}
\begin{minipage}{.3\textwidth}
\centering\includegraphics[width=2.35in]{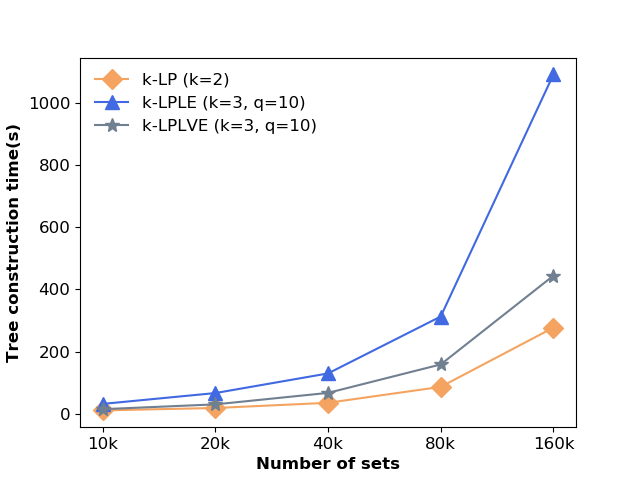}
\caption{Effects of increasing the number of sets on average number of questions (top) and tree construction time in seconds (bottom)}
\label{fig:effects-collection-size}
\end{minipage}
\end{figure*}


\subsubsection{Scalability with the number of entities and the collection size}
To evaluate the scalability of our algorithms on larger datasets, we conducted some experiments using our synthetic data. In one experiment, we varied the number of distinct entities in a collection, while keeping the number of sets and the overlap ratio fixed at 10k and 0.9 respectively. The number of distinct entities changes (as shown in Table~\ref{tab:synthetic-d}) with the set size varied.
As can be seen in Fig.~\ref{fig:effects-set-size}, the average number of questions is not affected much, but the tree construction time increases because of the larger number of candidate entities that are considered during the lower bounds calculation. The increase in running time is linear for $k$-LPLE and $k$-LPLVE, and that of $k$-LP is quadratic with $k=2$.

In another experiment, we varied the number of sets in the collection while keeping the set size in $[50,60]$ and the overlap ratio fixed at 0.9. The number of distinct entities $m$ increases as well (as shown in Table~\ref{tab:synthetic-n}), when we increase the number of sets $n$.
As shown in Fig.~\ref{fig:effects-collection-size}, with each doubling of the input size, the average number of questions increases roughly by 1. The tree construction time is expected to increase linearly with the number of sets if the number of distinct entities is fixed. In our experiment, the tree construction time looks a bit far from linear (and more quadratic) because of the increase in $m$ as $n$ increases.

\subsubsection{Evaluation results on query discovery}
Fig.~\ref{fig:query-discovery-questions-baseball} shows both the number of questions and the query discovery time to discover the target queries on the baseball database for the baseline InfoGain and our lookahead strategies. It can be seen that the number of questions for $k$-LP, $k$-LPLE, and $k$-LPLVE is less than or equal to InfoGain (except T7 for $k$-LP). Since none of the strategies are optimal, our strategies may sometimes require more questions than InfoGain, but that probability is very low as discussed in Section~\ref{sec:exp-set-discovery}.
Moreover, although the query discovery time of our strategies is higher than InfoGain, it is relatively small when the candidate queries have large result sets (on average 7000 to 12000), as shown in Table~\ref{tab:example-tuples-baseball}.
Finally, an important observation can be made about our query discovery strategy. The user is required to confirm the membership of only a few tuples (9 to 11) to find the target query among a large number of candidate queries (600 to 1200) which is more convenient than listing all the possible output tuples (close to 2000 for some of our target queries) of a target query.

\begin{figure*}[tb]
    \begin{subfigure}[b]{0.49\linewidth}
        \centering
        \begin{tabular}{c|c|c|c|c}
        \thead{Target query} & \thead{InfoGain} & \thead{$k$-LP\\($k = 2$)} & \thead{$k$-LPLE\\($k = 3$, $q = 10$)} &
        \thead{$k$-LPLVE\\($k = 3$, $q = 10$)}\\
        \hline
        T1 & 10 & 10 & 10 & 10\\
        T2 & 10 & 9 & 10 & 10\\
        T3 & 10 & 10 & 9 & 9\\
        T4 & 10 & 10 & 9 & 9\\
        T5 & 11 & 11 & 10 & 10\\
        T6 & 10 & 9 & 9 & 9\\
        T7 & 10 & 11 & 10 & 10\\
        
        \end{tabular}
        \caption{Number of questions}
        \label{tab:questions-baseball}
    \end{subfigure}
    \begin{subfigure}[b]{0.49\linewidth}
        \vspace*{5mm}
        \centering
       \begin{tabular}{c|c|c|c|c}
       \thead{Target query} & \thead{InfoGain} & \thead{$k$-LP\\($k = 2$)} & \thead{$k$-LPLE\\($k = 3$, $q = 10$)} &
        \thead{$k$-LPLVE\\($k = 3$, $q = 10$)}\\
        \hline
        T1 & 1.798 & 163.097 & 11.662 & 7.999\\
        T2 & 3.234 & 17.880 & 37.867 & 26.060\\
        T3 & 2.921 & 31.499 & 31.589 & 19.453\\
        T4 & 2.796 & 20.548 & 20.944 & 15.894\\
        T5 & 3.687 & 19.124 & 23.314 & 18.690\\
        T6 & 0.906 & 10.747 & 10.395 & 4.806\\
        T7 & 2.187 & 7.108 & 16.257 & 17.685\\
        \end{tabular}
        \caption{Query discovery time (seconds)}
        \label{tab:discovery-time-baseball}
    \end{subfigure}
    \caption{Number of questions and query discovery time to find the target queries on baseball database}
    \label{fig:query-discovery-questions-baseball}
\end{figure*}

\revised{
\section{Discussions}
\label{sec:discussion}
Our work is focused on reducing the number of interactions by selecting examples that are most informative, effectively reducing the number of candidates, but that is only one factor affecting the user experience. There are multiple other factors that need to be considered when applying our work in real settings.

\noindent \textbf{Multiple-choice examples}
  Sometimes it is more desirable to offer a set of examples (instead of one) and asking if one or more of those examples belong to the target set. For example, this can be more effective if the user is not sure about some examples. A interesting question is how those examples should be selected. One approach is to aim for maximizing the expected gain, as done in a multi-armed bandit setting. This can dramatically increases the size of the search space though, and using effective pruning strategies is essential. An alternative is to find some strategies for selecting the nodes of a decision tree that provide `good' sets of examples but not computationally intensive.

\noindent \textbf{Possibility of errors in answers}
  It is possible that users make mistakes in their answers, and this can introduce another interesting challenge in detecting that a mistake is made and recovering from them. One approach is to backtrack when no target set satisfies all constraints and revisit those constraints. An alternative is to assign a level of certainty, and make the optimization process aware of the uncertainties.
  
\noindent \textbf{Unanswered questions}}
  Sometimes the user is uncertain about the membership of an entity in the target set and may reply ``don't know'' to the membership question. In such cases, the entity selection strategy can be called again using the same collection of candidate sets but excluding the entities that the user is not sure about. With unanswered questions, the search may not resolve to a single set. 

\section{Conclusions}
\label{sec:conclusions}
We have studied the problem of set discovery using an interactive approach, where example entities from candidate sets are presented and the search is narrowed down based on the feedback about the presence of those entities in the target set. 
We have formulated the search as a tree optimization and have developed both effective and efficient k-step lookahead algorithms to construct a tree which results in near-optimal number of questions needed to discover a set. Our evaluation on both real and synthetic data shows the efficiency and scalability of our algorithms.

Our work can be extended or improved in a few directions, in addition to those highlighted in Section~\ref{sec:discussion}. One direction is to further study the distribution of entities in a collection. Better understanding the distribution may provide some insight to develop other strategies. Another direction is to study scenarios where the sets to be discovered are not equally likely. Extending our algorithms to the cases where the sets are noisy
or have errors
is another direction.



\section*{Acknowledgments}
	This research was funded by the Natural Sciences and Engineering Research Council of Canada and through a grant from Servus Credit Union.

\bibliographystyle{ACM-Reference-Format}
\bibliography{references}


\begin{thebibliography}{39}


\ifx \showCODEN    \undefined \def \showCODEN     #1{\unskip}     \fi
\ifx \showDOI      \undefined \def \showDOI       #1{#1}\fi
\ifx \showISBNx    \undefined \def \showISBNx     #1{\unskip}     \fi
\ifx \showISBNxiii \undefined \def \showISBNxiii  #1{\unskip}     \fi
\ifx \showISSN     \undefined \def \showISSN      #1{\unskip}     \fi
\ifx \showLCCN     \undefined \def \showLCCN      #1{\unskip}     \fi
\ifx \shownote     \undefined \def \shownote      #1{#1}          \fi
\ifx \showarticletitle \undefined \def \showarticletitle #1{#1}   \fi
\ifx \showURL      \undefined \def \showURL       {\relax}        \fi
\providecommand\bibfield[2]{#2}
\providecommand\bibinfo[2]{#2}
\providecommand\natexlab[1]{#1}
\providecommand\showeprint[2][]{arXiv:#2}

\bibitem[\protect\citeauthoryear{??}{SDS}{[n.d.]}]%
        {SDSS}
 \bibinfo{year}{[n.d.]}\natexlab{}.
\newblock \bibinfo{title}{Sloan Digital Sky Survey}.
\newblock
\newblock
\urldef\tempurl%
\url{https://www.sdss.org/}
\showURL{%
\tempurl}


\bibitem[\protect\citeauthoryear{??}{SQL}{[n.d.]}]%
        {SQLShare}
 \bibinfo{year}{[n.d.]}\natexlab{}.
\newblock \bibinfo{title}{SQLShare: Database-as-a-Service for Science}.
\newblock
\newblock
\urldef\tempurl%
\url{https://uwescience.github.io/sqlshare/}
\showURL{%
\tempurl}


\bibitem[\protect\citeauthoryear{Abouzied, Angluin, Papadimitriou, Hellerstein,
  and Silberschatz}{Abouzied et~al\mbox{.}}{2013}]%
        {Abouzied-interactive}
\bibfield{author}{\bibinfo{person}{Azza Abouzied}, \bibinfo{person}{Dana
  Angluin}, \bibinfo{person}{Christos Papadimitriou}, \bibinfo{person}{Joseph
  Hellerstein}, {and} \bibinfo{person}{Avi Silberschatz}.}
  \bibinfo{year}{2013}\natexlab{}.
\newblock \showarticletitle{Learning and Verifying Quantified Boolean Queries
  by Example}.
\newblock \bibinfo{journal}{\emph{Proceedings of the {PODS} Conference}}
  (\bibinfo{date}{04} \bibinfo{year}{2013}).
\newblock
\urldef\tempurl%
\url{https://doi.org/10.1145/2463664.2465220}
\showDOI{\tempurl}


\bibitem[\protect\citeauthoryear{Adler and Heeringa}{Adler and
  Heeringa}{2008}]%
        {Adler2008}
\bibfield{author}{\bibinfo{person}{Micah Adler} {and} \bibinfo{person}{Brent
  Heeringa}.} \bibinfo{year}{2008}\natexlab{}.
\newblock \showarticletitle{Approximating Optimal Binary Decision Trees}. In
  \bibinfo{booktitle}{\emph{Approximation, Randomization and Combinatorial
  Optimization. Algorithms and Techniques}},
  \bibfield{editor}{\bibinfo{person}{Ashish Goel}, \bibinfo{person}{Klaus
  Jansen}, \bibinfo{person}{Jos{\'e} D.~P. Rolim}, {and}
  \bibinfo{person}{Ronitt Rubinfeld}} (Eds.). \bibinfo{publisher}{Springer
  Berlin Heidelberg}, \bibinfo{address}{Berlin, Heidelberg},
  \bibinfo{pages}{1--9}.
\newblock
\showISBNx{978-3-540-85363-3}


\bibitem[\protect\citeauthoryear{Angluin}{Angluin}{1988}]%
        {angluin1988queries}
\bibfield{author}{\bibinfo{person}{Dana Angluin}.}
  \bibinfo{year}{1988}\natexlab{}.
\newblock \showarticletitle{Queries and concept learning}.
\newblock \bibinfo{journal}{\emph{Machine learning}} \bibinfo{volume}{2},
  \bibinfo{number}{4} (\bibinfo{year}{1988}), \bibinfo{pages}{319--342}.
\newblock


\bibitem[\protect\citeauthoryear{Arenas, Diaz, and Kostylev}{Arenas
  et~al\mbox{.}}{2016}]%
        {arenas2016reverse}
\bibfield{author}{\bibinfo{person}{Marcelo Arenas}, \bibinfo{person}{Gonzalo~I
  Diaz}, {and} \bibinfo{person}{Egor~V Kostylev}.}
  \bibinfo{year}{2016}\natexlab{}.
\newblock \showarticletitle{Reverse engineering SPARQL queries}. In
  \bibinfo{booktitle}{\emph{Proceedings of the {WWW} Conference}}.
  \bibinfo{pages}{239--249}.
\newblock


\bibitem[\protect\citeauthoryear{Basu~Roy, Wang, Das, Nambiar, and
  Mohania}{Basu~Roy et~al\mbox{.}}{2008}]%
        {Roy2008}
\bibfield{author}{\bibinfo{person}{Senjuti Basu~Roy}, \bibinfo{person}{Haidong
  Wang}, \bibinfo{person}{Gautam Das}, \bibinfo{person}{Ullas Nambiar}, {and}
  \bibinfo{person}{Mukesh Mohania}.} \bibinfo{year}{2008}\natexlab{}.
\newblock \showarticletitle{Minimum-Effort Driven Dynamic Faceted Search in
  Structured Databases}. In \bibinfo{booktitle}{\emph{Proceedings of the 17th
  ACM Conference on Information and Knowledge Management}} (Napa Valley,
  California, USA) \emph{(\bibinfo{series}{CIKM '08})}.
  \bibinfo{publisher}{Association for Computing Machinery},
  \bibinfo{address}{New York, NY, USA}, \bibinfo{pages}{13–22}.
\newblock
\showISBNx{9781595939913}
\urldef\tempurl%
\url{https://doi.org/10.1145/1458082.1458088}
\showDOI{\tempurl}


\bibitem[\protect\citeauthoryear{Bonifati, Ciucanu, and Stawork}{Bonifati
  et~al\mbox{.}}{2014}]%
        {Bonifati14interactiveinference}
\bibfield{author}{\bibinfo{person}{Angela Bonifati}, \bibinfo{person}{Radu
  Ciucanu}, {and} \bibinfo{person}{Slawomir Stawork}.}
  \bibinfo{year}{2014}\natexlab{}.
\newblock \showarticletitle{Interactive inference of join queries}. In
  \bibinfo{booktitle}{\emph{In EDBT}}. \bibinfo{pages}{451--462}.
\newblock


\bibitem[\protect\citeauthoryear{Bonifati, Ciucanu, and Staworko}{Bonifati
  et~al\mbox{.}}{2016}]%
        {Bonifati2016}
\bibfield{author}{\bibinfo{person}{Angela Bonifati}, \bibinfo{person}{Radu
  Ciucanu}, {and} \bibinfo{person}{S\l{}awek Staworko}.}
  \bibinfo{year}{2016}\natexlab{}.
\newblock \showarticletitle{Learning Join Queries from User Examples}.
\newblock \bibinfo{journal}{\emph{ACM Trans. Database Syst.}}
  \bibinfo{volume}{40}, \bibinfo{number}{4}, Article \bibinfo{articleno}{24}
  (\bibinfo{date}{Jan.} \bibinfo{year}{2016}), \bibinfo{numpages}{38}~pages.
\newblock
\showISSN{0362-5915}
\urldef\tempurl%
\url{https://doi.org/10.1145/2818637}
\showDOI{\tempurl}


\bibitem[\protect\citeauthoryear{Chatzopoulou, Eirinaki, and
  Polyzotis}{Chatzopoulou et~al\mbox{.}}{2009}]%
        {chatzopoulou2009query}
\bibfield{author}{\bibinfo{person}{Gloria Chatzopoulou},
  \bibinfo{person}{Magdalini Eirinaki}, {and} \bibinfo{person}{Neoklis
  Polyzotis}.} \bibinfo{year}{2009}\natexlab{}.
\newblock \showarticletitle{Query recommendations for interactive database
  exploration}. In \bibinfo{booktitle}{\emph{International Conference on
  Scientific and Statistical Database Management}}. Springer,
  \bibinfo{pages}{3--18}.
\newblock


\bibitem[\protect\citeauthoryear{Dimitriadou, Papaemmanouil, and
  Diao}{Dimitriadou et~al\mbox{.}}{2014}]%
        {Dimitriadou2014}
\bibfield{author}{\bibinfo{person}{Kyriaki Dimitriadou}, \bibinfo{person}{Olga
  Papaemmanouil}, {and} \bibinfo{person}{Yanlei Diao}.}
  \bibinfo{year}{2014}\natexlab{}.
\newblock \showarticletitle{Explore-by-Example: An Automatic Query Steering
  Framework for Interactive Data Exploration}. In
  \bibinfo{booktitle}{\emph{Proceedings of the 2014 ACM SIGMOD International
  Conference on Management of Data}} (Snowbird, Utah, USA)
  \emph{(\bibinfo{series}{SIGMOD '14})}. \bibinfo{publisher}{Association for
  Computing Machinery}, \bibinfo{address}{New York, NY, USA},
  \bibinfo{pages}{517–528}.
\newblock
\showISBNx{9781450323765}
\urldef\tempurl%
\url{https://doi.org/10.1145/2588555.2610523}
\showDOI{\tempurl}


\bibitem[\protect\citeauthoryear{Dinur and Steurer}{Dinur and Steurer}{2014}]%
        {dinur2014analytical}
\bibfield{author}{\bibinfo{person}{Irit Dinur} {and} \bibinfo{person}{David
  Steurer}.} \bibinfo{year}{2014}\natexlab{}.
\newblock \showarticletitle{Analytical approach to parallel repetition}. In
  \bibinfo{booktitle}{\emph{Proceedings of the forty-sixth annual ACM symposium
  on Theory of computing}}. \bibinfo{pages}{624--633}.
\newblock


\bibitem[\protect\citeauthoryear{Dong, Berti-Equille, and Srivastava}{Dong
  et~al\mbox{.}}{2009}]%
        {dong2009truth}
\bibfield{author}{\bibinfo{person}{Xin~Luna Dong}, \bibinfo{person}{Laure
  Berti-Equille}, {and} \bibinfo{person}{Divesh Srivastava}.}
  \bibinfo{year}{2009}\natexlab{}.
\newblock \showarticletitle{Truth discovery and copying detection in a dynamic
  world}.
\newblock \bibinfo{journal}{\emph{Proceedings of the VLDB Endowment}}
  \bibinfo{volume}{2}, \bibinfo{number}{1} (\bibinfo{year}{2009}),
  \bibinfo{pages}{562--573}.
\newblock


\bibitem[\protect\citeauthoryear{Esmeir and Markovitch}{Esmeir and
  Markovitch}{2004}]%
        {Esmeir2004}
\bibfield{author}{\bibinfo{person}{Saher Esmeir} {and} \bibinfo{person}{Shaul
  Markovitch}.} \bibinfo{year}{2004}\natexlab{}.
\newblock \showarticletitle{Lookahead-Based Algorithms for Anytime Induction of
  Decision Trees}. In \bibinfo{booktitle}{\emph{Proceedings of the Twenty-First
  International Conference on Machine Learning}} (Banff, Alberta, Canada)
  \emph{(\bibinfo{series}{ICML '04})}. \bibinfo{publisher}{Association for
  Computing Machinery}, \bibinfo{address}{New York, NY, USA},
  \bibinfo{pages}{33}.
\newblock
\showISBNx{1581138385}
\urldef\tempurl%
\url{https://doi.org/10.1145/1015330.1015373}
\showDOI{\tempurl}


\bibitem[\protect\citeauthoryear{Hasnat}{Hasnat}{2021}]%
        {Hasnat2021Thesis}
\bibfield{author}{\bibinfo{person}{Arif Hasnat}.}
  \bibinfo{year}{2021}\natexlab{}.
\newblock \emph{\bibinfo{title}{Interactive set discovery}}.
\newblock \bibinfo{thesistype}{Master's\ thesis}. \bibinfo{school}{University
  of Alberta}.
\newblock
\urldef\tempurl%
\url{https://doi.org/10.7939/r3-k9jr-am91}
\showDOI{\tempurl}


\bibitem[\protect\citeauthoryear{Howe, Cole, Souroush, Koutris, Key,
  Khoussainova, and Battle}{Howe et~al\mbox{.}}{2011}]%
        {howe2011database}
\bibfield{author}{\bibinfo{person}{Bill Howe}, \bibinfo{person}{Garret Cole},
  \bibinfo{person}{Emad Souroush}, \bibinfo{person}{Paraschos Koutris},
  \bibinfo{person}{Alicia Key}, \bibinfo{person}{Nodira Khoussainova}, {and}
  \bibinfo{person}{Leilani Battle}.} \bibinfo{year}{2011}\natexlab{}.
\newblock \showarticletitle{Database-as-a-service for long-tail science}. In
  \bibinfo{booktitle}{\emph{International Conference on Scientific and
  Statistical Database Management}}. Springer, \bibinfo{pages}{480--489}.
\newblock


\bibitem[\protect\citeauthoryear{Hyafil and Rivest}{Hyafil and Rivest}{1976}]%
        {HYAFIL1976}
\bibfield{author}{\bibinfo{person}{Laurent Hyafil} {and}
  \bibinfo{person}{Ronald~L. Rivest}.} \bibinfo{year}{1976}\natexlab{}.
\newblock \showarticletitle{Constructing optimal binary decision trees is
  NP-complete}.
\newblock \bibinfo{journal}{\emph{Inform. Process. Lett.}} \bibinfo{volume}{5},
  \bibinfo{number}{1} (\bibinfo{year}{1976}), \bibinfo{pages}{15 -- 17}.
\newblock
\showISSN{0020-0190}
\urldef\tempurl%
\url{https://doi.org/10.1016/0020-0190(76)90095-8}
\showDOI{\tempurl}


\bibitem[\protect\citeauthoryear{Jain, Moritz, Halperin, Howe, and
  Lazowska}{Jain et~al\mbox{.}}{2016}]%
        {jain2016sqlshare}
\bibfield{author}{\bibinfo{person}{Shrainik Jain}, \bibinfo{person}{Dominik
  Moritz}, \bibinfo{person}{Daniel Halperin}, \bibinfo{person}{Bill Howe},
  {and} \bibinfo{person}{Ed Lazowska}.} \bibinfo{year}{2016}\natexlab{}.
\newblock \showarticletitle{Sqlshare: Results from a multi-year
  sql-as-a-service experiment}. In \bibinfo{booktitle}{\emph{Proceedings of the
  2016 International Conference on Management of Data}}.
  \bibinfo{pages}{281--293}.
\newblock


\bibitem[\protect\citeauthoryear{Kalashnikov, Lakshmanan, and
  Srivastava}{Kalashnikov et~al\mbox{.}}{2018}]%
        {kalashnikov2018fastqre}
\bibfield{author}{\bibinfo{person}{Dmitri~V Kalashnikov},
  \bibinfo{person}{Laks~VS Lakshmanan}, {and} \bibinfo{person}{Divesh
  Srivastava}.} \bibinfo{year}{2018}\natexlab{}.
\newblock \showarticletitle{Fastqre: Fast query reverse engineering}. In
  \bibinfo{booktitle}{\emph{Proceedings of the 2018 International Conference on
  Management of Data}}. \bibinfo{pages}{337--350}.
\newblock


\bibitem[\protect\citeauthoryear{Khoussainova, Kwon, Liao, Balazinska,
  Gatterbauer, and Suciu}{Khoussainova et~al\mbox{.}}{2011}]%
        {khoussainova2011session}
\bibfield{author}{\bibinfo{person}{Nodira Khoussainova},
  \bibinfo{person}{YongChul Kwon}, \bibinfo{person}{Wei-Ting Liao},
  \bibinfo{person}{Magdalena Balazinska}, \bibinfo{person}{Wolfgang
  Gatterbauer}, {and} \bibinfo{person}{Dan Suciu}.}
  \bibinfo{year}{2011}\natexlab{}.
\newblock \showarticletitle{Session-based browsing for more effective query
  reuse}. In \bibinfo{booktitle}{\emph{International Conference on Scientific
  and Statistical Database Management}}. Springer, \bibinfo{pages}{583--585}.
\newblock


\bibitem[\protect\citeauthoryear{Kumar, Raghavan, Rajagopalan, Sivakumar,
  Tomkins, and Upfal}{Kumar et~al\mbox{.}}{2000}]%
        {kumar2000stochastic}
\bibfield{author}{\bibinfo{person}{Ravi Kumar}, \bibinfo{person}{Prabhakar
  Raghavan}, \bibinfo{person}{Sridhar Rajagopalan}, \bibinfo{person}{D
  Sivakumar}, \bibinfo{person}{Andrew Tomkins}, {and} \bibinfo{person}{Eli
  Upfal}.} \bibinfo{year}{2000}\natexlab{}.
\newblock \showarticletitle{Stochastic models for the web graph}. In
  \bibinfo{booktitle}{\emph{Proceedings 41st Annual Symposium on Foundations of
  Computer Science}}. IEEE, \bibinfo{pages}{57--65}.
\newblock


\bibitem[\protect\citeauthoryear{Lahman}{Lahman}{2020}]%
        {baseball}
\bibfield{author}{\bibinfo{person}{Sean Lahman}.}
  \bibinfo{year}{2020}\natexlab{}.
\newblock \bibinfo{title}{Baseball database}.
\newblock
\newblock
\urldef\tempurl%
\url{http://www.seanlahman.com/baseball-archive/statistics/}
\showURL{%
\tempurl}


\bibitem[\protect\citeauthoryear{Laurent and Rivest}{Laurent and
  Rivest}{1976}]%
        {laurent1976constructing}
\bibfield{author}{\bibinfo{person}{Hyafil Laurent} {and}
  \bibinfo{person}{Ronald~L Rivest}.} \bibinfo{year}{1976}\natexlab{}.
\newblock \showarticletitle{Constructing optimal binary decision trees is
  NP-complete}.
\newblock \bibinfo{journal}{\emph{Information processing letters}}
  \bibinfo{volume}{5}, \bibinfo{number}{1} (\bibinfo{year}{1976}),
  \bibinfo{pages}{15--17}.
\newblock


\bibitem[\protect\citeauthoryear{Li, Chan, and Maier}{Li et~al\mbox{.}}{2015}]%
        {Li2015}
\bibfield{author}{\bibinfo{person}{Hao Li}, \bibinfo{person}{Chee-Yong Chan},
  {and} \bibinfo{person}{David Maier}.} \bibinfo{year}{2015}\natexlab{}.
\newblock \showarticletitle{Query from Examples: An Iterative, Data-Driven
  Approach to Query Construction}.
\newblock \bibinfo{journal}{\emph{Proc. VLDB Endow.}} \bibinfo{volume}{8},
  \bibinfo{number}{13} (\bibinfo{date}{Sept.} \bibinfo{year}{2015}),
  \bibinfo{pages}{2158–2169}.
\newblock
\showISSN{2150-8097}
\urldef\tempurl%
\url{https://doi.org/10.14778/2831360.2831369}
\showDOI{\tempurl}


\bibitem[\protect\citeauthoryear{Milo and Somech}{Milo and Somech}{2018}]%
        {milo2018next}
\bibfield{author}{\bibinfo{person}{Tova Milo} {and} \bibinfo{person}{Amit
  Somech}.} \bibinfo{year}{2018}\natexlab{}.
\newblock \showarticletitle{Next-step suggestions for modern interactive data
  analysis platforms}. In \bibinfo{booktitle}{\emph{Proceedings of the 24th ACM
  SIGKDD International Conference on Knowledge Discovery \& Data Mining}}.
  \bibinfo{pages}{576--585}.
\newblock


\bibitem[\protect\citeauthoryear{Mottin, Lissandrini, Velegrakis, and
  Palpanas}{Mottin et~al\mbox{.}}{2014}]%
        {mottin2014exemplar}
\bibfield{author}{\bibinfo{person}{Davide Mottin}, \bibinfo{person}{Matteo
  Lissandrini}, \bibinfo{person}{Yannis Velegrakis}, {and}
  \bibinfo{person}{Themis Palpanas}.} \bibinfo{year}{2014}\natexlab{}.
\newblock \showarticletitle{Exemplar queries: Give me an example of what you
  need}.
\newblock \bibinfo{journal}{\emph{Proceedings of the VLDB Endowment}}
  \bibinfo{volume}{7}, \bibinfo{number}{5} (\bibinfo{year}{2014}),
  \bibinfo{pages}{365--376}.
\newblock


\bibitem[\protect\citeauthoryear{Mottin, Lissandrini, Velegrakis, and
  Palpanas}{Mottin et~al\mbox{.}}{2017}]%
        {mottin2017new}
\bibfield{author}{\bibinfo{person}{Davide Mottin}, \bibinfo{person}{Matteo
  Lissandrini}, \bibinfo{person}{Yannis Velegrakis}, {and}
  \bibinfo{person}{Themis Palpanas}.} \bibinfo{year}{2017}\natexlab{}.
\newblock \showarticletitle{New trends on exploratory methods for data
  analytics}.
\newblock \bibinfo{journal}{\emph{Proceedings of the VLDB Endowment}}
  \bibinfo{volume}{10}, \bibinfo{number}{12} (\bibinfo{year}{2017}),
  \bibinfo{pages}{1977--1980}.
\newblock


\bibitem[\protect\citeauthoryear{Quinlan}{Quinlan}{1993}]%
        {Quinlan1993}
\bibfield{author}{\bibinfo{person}{{J. R.} Quinlan}.}
  \bibinfo{year}{1993}\natexlab{}.
\newblock \bibinfo{booktitle}{\emph{C4.5: Programs for machine learning}}.
\newblock \bibinfo{publisher}{Morgan Kaufmann}, \bibinfo{address}{San Mateo,
  CA}.
\newblock


\bibitem[\protect\citeauthoryear{Quinlan}{Quinlan}{1986}]%
        {Quinlan1986}
\bibfield{author}{\bibinfo{person}{J.~Ross Quinlan}.}
  \bibinfo{year}{1986}\natexlab{}.
\newblock \showarticletitle{Induction of Decision Trees}.
\newblock \bibinfo{journal}{\emph{Mach. Learn.}} \bibinfo{volume}{1},
  \bibinfo{number}{1} (\bibinfo{date}{March} \bibinfo{year}{1986}),
  \bibinfo{pages}{81–106}.
\newblock
\showISSN{0885-6125}
\urldef\tempurl%
\url{https://doi.org/10.1023/A:1022643204877}
\showDOI{\tempurl}


\bibitem[\protect\citeauthoryear{Settles}{Settles}{2009}]%
        {settles2009active}
\bibfield{author}{\bibinfo{person}{Burr Settles}.}
  \bibinfo{year}{2009}\natexlab{}.
\newblock \showarticletitle{Active learning literature survey}.
\newblock  (\bibinfo{year}{2009}).
\newblock


\bibitem[\protect\citeauthoryear{Sieling}{Sieling}{2008}]%
        {sieling2008minimization}
\bibfield{author}{\bibinfo{person}{Detlef Sieling}.}
  \bibinfo{year}{2008}\natexlab{}.
\newblock \showarticletitle{Minimization of decision trees is hard to
  approximate}.
\newblock \bibinfo{journal}{\emph{J. Comput. System Sci.}}
  \bibinfo{volume}{74}, \bibinfo{number}{3} (\bibinfo{year}{2008}),
  \bibinfo{pages}{394--403}.
\newblock


\bibitem[\protect\citeauthoryear{Tran, Chan, and Parthasarathy}{Tran
  et~al\mbox{.}}{2009}]%
        {Tran2009QBO}
\bibfield{author}{\bibinfo{person}{Quoc~Trung Tran}, \bibinfo{person}{Chee-Yong
  Chan}, {and} \bibinfo{person}{Srinivasan Parthasarathy}.}
  \bibinfo{year}{2009}\natexlab{}.
\newblock \showarticletitle{Query by Output}. In
  \bibinfo{booktitle}{\emph{Proceedings of the 2009 ACM SIGMOD International
  Conference on Management of Data}} (Providence, Rhode Island, USA)
  \emph{(\bibinfo{series}{SIGMOD '09})}. \bibinfo{publisher}{Association for
  Computing Machinery}, \bibinfo{address}{New York, NY, USA},
  \bibinfo{pages}{535–548}.
\newblock
\showISBNx{9781605585512}
\urldef\tempurl%
\url{https://doi.org/10.1145/1559845.1559902}
\showDOI{\tempurl}


\bibitem[\protect\citeauthoryear{Tran, Chan, and Parthasarathy}{Tran
  et~al\mbox{.}}{2014}]%
        {tran2014query}
\bibfield{author}{\bibinfo{person}{Quoc~Trung Tran},
  \bibinfo{person}{Chee\-Yong Chan}, {and} \bibinfo{person}{Srinivasan
  Parthasarathy}.} \bibinfo{year}{2014}\natexlab{}.
\newblock \showarticletitle{Query Reverse Engineering}.
\newblock \bibinfo{journal}{\emph{The VLDB Journal}} \bibinfo{volume}{23},
  \bibinfo{number}{5} (\bibinfo{date}{Oct.} \bibinfo{year}{2014}),
  \bibinfo{pages}{721–746}.
\newblock
\showISSN{1066-8888}
\urldef\tempurl%
\url{https://doi.org/10.1007/s00778-013-0349-3}
\showDOI{\tempurl}


\bibitem[\protect\citeauthoryear{Weiss and Cohen}{Weiss and Cohen}{2017}]%
        {Weiss-reverse-engineering}
\bibfield{author}{\bibinfo{person}{Yaacov~Y. Weiss} {and} \bibinfo{person}{Sara
  Cohen}.} \bibinfo{year}{2017}\natexlab{}.
\newblock \showarticletitle{Reverse Engineering SPJ-Queries from Examples}. In
  \bibinfo{booktitle}{\emph{Proceedings of the 36th ACM SIGMOD-SIGACT-SIGAI
  Symposium on Principles of Database Systems}} (Chicago, Illinois, USA)
  \emph{(\bibinfo{series}{PODS '17})}. \bibinfo{publisher}{Association for
  Computing Machinery}, \bibinfo{address}{New York, NY, USA},
  \bibinfo{pages}{151–166}.
\newblock
\showISBNx{9781450341981}
\urldef\tempurl%
\url{https://doi.org/10.1145/3034786.3056112}
\showDOI{\tempurl}


\bibitem[\protect\citeauthoryear{Zhang, Lipton, Li, and Smola}{Zhang
  et~al\mbox{.}}{2021}]%
        {zhang2021dive}
\bibfield{author}{\bibinfo{person}{Aston Zhang}, \bibinfo{person}{Zachary~C
  Lipton}, \bibinfo{person}{Mu Li}, {and} \bibinfo{person}{Alexander~J Smola}.}
  \bibinfo{year}{2021}\natexlab{}.
\newblock \showarticletitle{Dive into deep learning}.
\newblock \bibinfo{journal}{\emph{arXiv preprint arXiv:2106.11342}}
  (\bibinfo{year}{2021}).
\newblock


\bibitem[\protect\citeauthoryear{Zhang, Elmeleegy, Procopiuc, and
  Srivastava}{Zhang et~al\mbox{.}}{2013}]%
        {Zhang2013}
\bibfield{author}{\bibinfo{person}{Meihui Zhang}, \bibinfo{person}{Hazem
  Elmeleegy}, \bibinfo{person}{Cecilia Procopiuc}, {and}
  \bibinfo{person}{Divesh Srivastava}.} \bibinfo{year}{2013}\natexlab{}.
\newblock \showarticletitle{Reverse engineering complex join queries}.
\newblock \bibinfo{journal}{\emph{Proceedings of the ACM SIGMOD International
  Conference on Management of Data}}, \bibinfo{pages}{809--820}.
\newblock
\urldef\tempurl%
\url{https://doi.org/10.1145/2463676.2465320}
\showDOI{\tempurl}


\bibitem[\protect\citeauthoryear{Zhang, Zhang, Sellam, and Wu}{Zhang
  et~al\mbox{.}}{2019}]%
        {zhang2019mining}
\bibfield{author}{\bibinfo{person}{Qianrui Zhang}, \bibinfo{person}{Haoci
  Zhang}, \bibinfo{person}{Thibault Sellam}, {and} \bibinfo{person}{Eugene
  Wu}.} \bibinfo{year}{2019}\natexlab{}.
\newblock \showarticletitle{Mining precision interfaces from query logs}. In
  \bibinfo{booktitle}{\emph{Proceedings of the 2019 International Conference on
  Management of Data}}. \bibinfo{pages}{988--1005}.
\newblock


\bibitem[\protect\citeauthoryear{Zhong, Yu, and Klein}{Zhong
  et~al\mbox{.}}{2020}]%
        {zhong-testSuite}
\bibfield{author}{\bibinfo{person}{Ruiqi Zhong}, \bibinfo{person}{Tao Yu},
  {and} \bibinfo{person}{Dan Klein}.} \bibinfo{year}{2020}\natexlab{}.
\newblock \showarticletitle{Semantic Evaluation for Text-to-SQL with Distilled
  Test Suites}. In \bibinfo{booktitle}{\emph{Proceedings of the 2020 Conference
  on Empirical Methods in Natural Language Processing (EMNLP)}}.
  \bibinfo{publisher}{Association for Computational Linguistics},
  \bibinfo{pages}{396--411}.
\newblock


\bibitem[\protect\citeauthoryear{Zloof}{Zloof}{1975}]%
        {zloof1975query}
\bibfield{author}{\bibinfo{person}{Mosh{\'e}~M Zloof}.}
  \bibinfo{year}{1975}\natexlab{}.
\newblock \showarticletitle{Query by example}. In
  \bibinfo{booktitle}{\emph{Proceedings of the national computer conference and
  exposition}}. \bibinfo{pages}{431--438}.
\newblock


\end{thebibliography}

\end{document}